\documentclass[orivec,runningheads]{llncs}
\usepackage[margin=1.4in]{geometry}

\usepackage{amsmath,amssymb}
\usepackage{graphicx}
\usepackage{xspace}
\usepackage{tikz}
\usepackage{todonotes}
\usepackage{subfigure}
\usetikzlibrary{arrows,shapes,shapes.multipart,snakes,automata,backgrounds,petri,positioning,shadows,matrix,decorations.pathmorphing,fit,positioning,calc,backgrounds}
\usepackage{paralist}
\usepackage{times}
\usepackage{comment}
\usepackage[algo2e,lined,ruled]{algorithm2e} 
	
\newcommand{\delem}[1]{#1}
\newcommand{\celem}[1]{#1^c}

\newcommand{\net}{DPN\xspace}
\newcommand{\netname}{Data Petri net\xspace}
\newcommand{\cpn}{CPN\xspace}
\newcommand{\anet}{\mathcal{N}}

\newcommand{\acpn}{\celem{\anet}}
\newcommand{\acpnrep}{\celem{\restrict{\anet}}}
\newcommand{\anetrep}{\restrict{\anet}}

\newcommand{\svassignment}{SV assignment\xspace}

\def\post#1{\ensuremath{{#1}\kern-.05ex\bullet}}

\def\posto#1{\ensuremath{{#1}\kern-.05ex\circ}}

\newcommand{\true}{\small{\textit{\texttt{true}}}\xspace}
\newcommand{\false}{\small{\textit{\texttt{false}}}\xspace}

\newcommand{\DecTrans}{T_{dec}}

\newcommand{\PS}{pwr}
\newcommand{\MS}{\mathbb{B}}

\newcommand{\Vars}{V}
\newcommand{\VarsR}{\Vars^r}
\newcommand{\VarsW}{\Vars^w}

\newcommand{\varname}{v}

\newcommand{\typevar}{\var_\Domain}
\newcommand{\TypeVars}{\mathcal{V}}

\newcommand{\var}{v}
\newcommand{\varr}{\var^r}
\newcommand{\varw}{\var^w}

\newcommand{\varDomAss}{\mathit{dom}}
\newcommand{\varInit}{\varState_I}
\newcommand{\reads}{\mathit{read}}
\newcommand{\writes}{\mathit{write}}
\newcommand{\guard}{\mathit{guard}}

\newcommand{\DPN}{(P,\allowbreak T,\allowbreak F,\allowbreak \Vars, \allowbreak \varDomAss, \allowbreak \varInit, \allowbreak \reads, \allowbreak\writes,\allowbreak \guard)}
\newcommand{\DPNDELEM}{(\delem{P},\allowbreak \delem{T},\allowbreak \delem{F},\allowbreak \delem{\Vars},\allowbreak \varDomAss,\allowbreak \delem{\varInit}, \allowbreak \delem{\reads}, \allowbreak\delem{\writes},\allowbreak \delem{\guard})}

\newcommand{\varState}{\alpha}

\newcommand{\bind}{\beta}
\newcommand{\cpnbind}{\gamma}

\newcommand{\set}[1]{\{#1\}}                      
\newcommand{\tup}[1]{\langle #1\rangle}            
\newcommand{\card}[1]{|{#1}|}

\newcommand{\goto}[1]{\mathrel{\raisebox{-2pt}{$\xrightarrow{#1}$}}}

\newcommand{\Domain}{\mathcal{D}}
\newcommand{\Universe}{\mathcal{U}}
\newcommand{\varDomain}{\Delta_\Domain}
\newcommand{\Preds}{\Sigma_\Domain}
\newcommand{\domainName}{domain\xspace}
\newcommand{\domainNames}{domains\xspace}

\newcommand{\Integers}{\mathbb{Z}}
\newcommand{\Reals}{\mathbb{R}}

\newcommand{\Expr}[0]{\Phi}

\newcommand{\Consts}[1]{\mathcal{C}_{#1}}

\newcommand{\restrict}[1]{{#1}_{\Repres{}}}
\newcommand{\varStateR}{\varState_{\Repres{}}}
\newcommand{\Repres}[1]{\bar{\Delta}_{#1}}

\newcommand{\Bools}{\mathsf{bool}}
\newcommand{\UNITs}{\bullet}
\newcommand{\Strings}{\mathsf{Strings}}

\renewenvironment{proof}{\paragraph{Proof.}}{\hfill$\square$}

\usepackage{color,soul}

\begin{document}

\sloppy

\title{A Holistic Approach for Soundness Verification \\of Decision-Aware Process Models\\(extended version)}
\titlerunning{Soundness Verification of Decision-Aware Process Models\\ (Extended version)}

\author{Massimiliano de Leoni\inst{1} \and
 Paolo Felli\inst{2} \and
 Marco Montali\inst{2}}

\authorrunning{M.~de Leoni, P.~Felli, M.~Montali}

\institute{
  Eindhoven University of Technology, the Netherlands
  \\\email{m.d.leoni@tue.nl}
\and
  Free University of Bozen-Bolzano,
  \\\email{\{pfelli,montali\}@inf.unibz.it}
}

\maketitle

\begin{abstract}
The last decade has witnessed an increasing transformation in the design, engineering, and mining of processes, moving from a pure control-flow perspective to more integrated models where also data and decisions are explicitly considered.
This calls for methods and techniques able to ascertain the correctness of such integrated models. Differently from previous approaches, which mainly focused on the local interplay between decisions and their corresponding outgoing branches, we introduce a holistic approach to verify the end-to-end soundness of a Petri net-based process model, enriched with case data and decisions.
In particular, we present an effective, implemented technique that verifies soundness by translating the input net into a colored Petri net with bounded color domains, which can then be analyzed using conventional tools. We prove correctness and termination of this technique. In addition, we relate our contribution to recent results on decision-aware soundness, showing that our approach can be readily applied there.
\end{abstract}


\section{Introduction}
  The fundamental problem of verifying the correctness of business process models has been traditionally tackled by exclusively considering the control flow perspective. This means that correctness is assessed by only considering the ordering relations among activities present in the model. In this setting, one of the most investigated formal notions of correctness is that of \emph{soundness}, originally introduced by van der Aalst in the context of workflow nets (a special class of Petri nets that is suitable to capture the control flow of business processes) \cite{Aalst1998workflow}. Intuitively, soundness guarantees the two good properties of ``possibility of clean termination" and of ``absence of deadlocks". On the one hand, it ensures that whenever a process instance is being executed, it always has the possibility of reaching the completion of the process, and if it does so, then no running concurrent thread is still active in the process. On the other hand, it captures that all parts of the process can be executed in some scenario, that is, the process does not contain dead activities that are impossible to enact.

  The control-flow perspective is certainly of high importance as it can be considered the main process backbone; however, many other perspectives should also be taken into account. In fact, the last decade has witnessed an increasing transformation in the design, engineering, and mining of processes, moving from a pure control-flow perspective to more integrated models where also data and decisions are explicitly considered. The fact that the incorporation of decisions within   process models is gaining momentum is also testified by the recent introduction and development of the Decision Model and Notation (DMN), an OMG standard~\cite{DMN.2011}. This calls for methods and techniques able to ascertain the correctness of such integrated models, which is important not only during the design phase of the business process lifecycle, but also when it comes to decision and guard mining~\cite{Leoni.2013b}, as well as compliance checking~\cite{KLR10}.

Previous approaches to analyze correctness of decision-aware processes have typically focused on single decisions \cite{CDL16} or on  the local interplay between decisions and their corresponding outgoing branches \cite{Batoulis2017}. More recent efforts have tackled this locality problem by holistically considering soundness of the overall, end-to-end process in the presence of data and decisions, but have mainly stayed at the foundational level \cite{MonC16,BaHW17}. In particular, they do not come with actual techniques to effectively carry out the verification of soundness. In this work, we overcome this limitation by introducing a holistic, formal and operational approach to verify the end-to-end soundness of \netname{s} (\net{s}) \cite{deLeoni2014DPN}. \net{s} combine workflow nets with case data, decisions and conditional data updates, achieving a suitable balance between expressiveness and simplicity. Thanks to their solid formal foundation, \net{s} come with a clear execution semantics, and consequently allow us to unambiguously extend the notion of soundness to incorporate the decision perspective. In addition, they combine the main ingredients that are needed to formally capture conventional process modelling notations, such as the combination of BPMN control- and data-flow with DMN decisions. 

In the general case, verifying soundness of \net{s} is undecidable, due to the presence of case data and the possibility of manipulating them so as to reconstruct Turing-powerful computational devices. This applies, in particular, when case data can be updated using arithmetical operators. We isolate here a decidable class of \net{s} that employs both non-numerical and numerical domains, and is expressive enough to capture data-aware process models equipped with S-FEEL DMN decisions \cite{DMN.2011}, such as those recently proposed in \cite{Batoulis2017,BaHW17}.
Importantly, such \net{s} cannot be directly analyzed algorithmically: due to the presence of data and corresponding updates, they in fact induce a state space that may have infinitely many states even when the control-flow is expressed by a bounded workflow net.  To tame this infinity, we take inspiration from the technique of predicate abstraction \cite{CGL94}, and in particular the approach adopted in \cite{KLR10}, and present an effective technique that verifies soundness by translating the input net into a colored Petri net (\cpn) with bounded color domains, which induces a finite state space and can be consequently analyzed using conventional tools. This technique has been implemented as a plug-in of the well-established ProM process mining framework.

The paper is organized as follows. In Section~\ref{sec:related} we discuss related work; in Section~\ref{sec:syntax-semantics} we provide the necessary background on \net{s} and a precise formalisation of its execution semantics; In Section~\ref{sec:dmn} we discuss the relation between \net{s} and the DMN S-FEEL language; in Section~\ref{sec:verification} we illustrate an effective technique for translating a given \net into a special colored Petri net with bounded color domains and which can be thus studied using standard tools, and then finally prove that we can analyze it to assess the properties of the original \net, including soundness. 
Section~\ref{sec:implementation} discusses the ProM implementation and reports on a number of experiments based on models of real-life processes, some of which were designed by hand and others were a combination of hand-design and of process discovery. Experiments show that the technique is operationally effective and  can be applied to real-life case studies. Finally, Section~\ref{sec:conclusion} concludes the paper and delineates avenues of future research.

\section{Related Work}
\label{sec:related}

Within the field of database theory, many approaches have been proposed to formalize and verify very sophisticated variants of data-aware processes \cite{CaDM13}, also considering data-aware extensions of soundness \cite{MonC16}. However, such works are mainly foundational, and do not currently come with effective verification algorithms and implementations.
Within the field of business process management and information systems engineering, a plethora of techniques and tools exists for verifying soundness of process models that only capture the control-flow perspective, but not much research has been carried out to incorporate the data and decision perspective in the analysis. Sadiq et al.~\cite{Sadiq:2004:DFV:1012294.1012317} were among the first to acknowledge the importance of incorporating the data perspective within soundness analysis, but they did not propose any technique to carry out the verification. Sidorova et al.\ proposed a conceptual extension of workflow nets, equipping them with an abstract, high-level data model~\cite{Sidorova:2010:WSR:1883784.1883835,Sidorova:2011:SVC:1994485.1995288}. In this approach, data are captured abstractly, and it is assumed that activities read and write entire guards, instead of reading and writing data variables that affect the satisfaction of guards. This abstract approach certainly simplifies the analysis because reading and writing guards is equivalent to reading and writing boolean values, which corresponds to a sort of a-priori propositionalization of the data.
This is, however,  not realistic: as testified by modern process modeling notations, such as BPMN and DMN, the data perspective requires (at least) to have data variables and full-fledged guards and updates over them.
\cite{CDL16} focuses on single DMN decision, in particular verifying whether a DMN table is correct, or contains instead inconsistencies, missing and overlapping rules. This certainly fits in the context of data-aware soundness,  but it is only a minor portion of it, since the analysis is only conducted locally to decision points in the process, and local forms of analysis do not suffice to guarantee good behavioral properties of the entire process \cite{KLR10}. A similar drawback is also present in~\cite{Batoulis2017}, where the contribution of decisions in the verification of soundness is also local, and limits itself to the interaction between decisions and their immediate outgoing sequence flows.
As mentioned in the introduction, soundness verification plays a key role in decision and guard mining~\cite{Leoni.2013b}. In this setting, an initial process model is discovered by solely considering the control-flow perspective. In a second phase, decision points present in the model are enriched with decisions and conditions again inferred from the event data present in the log. This ``local enrichment" does not guarantee that the overall model is indeed sound, so soundness verification techniques should be inserted in the loop to discard incorrect results and guide discovery towards the extraction of models that are correct by design.
%

The two closest works to our contribution are \cite{BaHW17,KLR10}. In \cite{BaHW17}, the authors consider the interplay between BPMN and DMN, providing different notions of data-aware soundness on top of such process models (once the BPMN component is encoded into a Petri net, which can be seamlessly tackled by known techniques \cite{DDO08,Kalenkova:2016:PMU:2976767.2987688}). As shown in Section~\ref{sec:dmn}, our approach is expressive enough to capture the process models studied in \cite{BaHW17}. In addition, our verification technique based on an encoding into \cpn{s} does not only preserve the notion of soundness we introduce, but it actually guarantees that the obtained \cpn is behaviorally equivalent to the input \net. This, in turn, implies that all variants of soundness defined in \cite{BaHW17} can be actually verified using this approach.

In \cite{KLR10}, the authors introduce an abstraction approach that shares the same spirit of our technique: it is faithful (i.e., it preserves properties), and it is based on the idea of considering only boundedly many representative values in place of the entire data domains. There are however four fundamental differences between our setting and that of \cite{KLR10}. First of all, in  \cite{KLR10} abstractions are used to shrink the state space of the analysis, while in our case they are employed to tame the infinity brought by the presence of data and the possibility of updating them. Second, \cite{KLR10} defines abstract process graphs that do not come with a formal execution semantics, and consequently do not allow one to formally prove that the abstraction technique is indeed correct. Since our approach is expressive enough to capture the model of \cite{KLR10} (see Section~\ref{sec:encoding}), our correctness result captured in Section~\ref{sec:translation-correctness} can be actually lifted to \cite{KLR10} as well. Third, \cite{KLR10} focuses on compliance checking against LTL-based compliance rules, which are unable to capture soundness (in particular, the ``possibility of termination", which has an intrinsic branching nature); on the other hand, since our encoding produces a \cpn that is behaviorally equivalent to the original \net, it also preserves all the runs and, in turn, all LTL properties. Finally, while \cite{KLR10} translates the problem of compliance checking into a temporal model checking problem, we resort to Petri net-based techniques. 

\section{Syntax and Semantics of \net{s}}
\label{sec:syntax-semantics}

We provide the necessary background on the \net model \cite{deLeoni2014DPN}, precisely defining its execution semantics and introducing a running example. We then lift the standard notion of soundness to the more sophisticated setting of \net{s}.
%
Assume an infinite universe of possible values $\Universe$.

\begin{definition}[Domain]
A \emph{domain} \domainName is a couple $\Domain=\tup{\varDomain,\Preds}$ where $\varDomain\subseteq \Universe$ is a  set of possible values and $\Preds$ is the set of binary predicates on $\varDomain$.
\end{definition}

We consider a set of domains $\mathfrak{D}$, and in particular the notable \domainNames $\Domain_{\Reals} = \tup{\Reals, \set{<,>,=,\neq}}$, $\Domain_{\Integers} = \tup{\Integers, \set{<,>,=,\neq}}$, $\Domain_{\mathit{bool}}=\tup{\set{True,False},\set{=,\neq}}$, $\Domain_{\mathit{string}}=\tup{\mathbb{S},\set{=}}$ which, respectively, account for real numbers, integers, booleans, and strings ($\mathbb{S}$ denotes here the infinite set of all strings).

Consider a set $\Vars$ of \emph{variables}.
Given a variable $\var\in\Vars$ write $\varr$ or $\varw$ to denote that the variable $\var$ is read or written, hence we consider two distinct sets $\VarsR$ and $\VarsW$ defined as $\VarsR = \set{ \var^r \mid \var\in\Vars }$ and $\VarsW = \set{ \var^w \mid \var\in\Vars }$.
When we do need to distinguish, we still use the symbol $\var$ to denote any member of $(\VarsR\cup\VarsW)$.
To talk about the possible values variables may assume, we need to associate domains to variables. If a variable $\var$ is assigned a domain $\Domain=\tup{\varDomain,\Preds}$, for brevity we denote by $\typevar$ the corresponding \emph{typed variable}, that is a shorthand to specify that $\var$ can only assume values in $\varDomain$.

Variables provide the basic building block to define logical conditions (formally, guards) on data.

\begin{definition}[Guards]
Given a set of typed variables $\TypeVars$ for a set $\Vars$, the set of possible \emph{guards} $\Expr(\Vars)$  is the largest set containing the following:
\begin{compactitem}
\item $\var$ iff $\var \in (\VarsR\cup\VarsW)$;
\item $\typevar \odot \varDomain$  iff $\var \in (\VarsR\cup\VarsW)$ and $\odot \in \Preds$;
\item $\phi_1 \land \phi_2$ and $\phi_1 \lor \phi_2$ iff $\phi_1$ and $\phi_2$ are guards in $\Expr(\Vars)$.
\end{compactitem}
\label{def:guards}
\end{definition}

A \emph{variable assignment} is a function $\bind :  (\VarsR\cup\VarsW) \rightarrow \Universe\cup\set{\bot}$, which assigns a value to read and written variables, with the restriction that $\bind(\var)$ is a possible value for $\var$, that is if $\typevar$ is the corresponding typed variable then $\bind(\var)\in\varDomain$. The symbol $\bot$ is used to denote an undefined value, i.e., that the variable is not set.
%
Given a variable assignment $\bind$ and a guard $\phi$, we say that $\phi$ evaluates to true when variables are substituted as per $\bind$, written $\phi_{[\bind]}=\true$, iff:
\begin{compactitem}
\item if $\phi=\var$ then $\bot\neq\bind(\var)$;
\item if $\phi=\var \odot k$, then $\odot(x , k)$ for $x =\bind(\var)$;
\item if $\phi = \phi_1 \land \phi_2$ then $\phi_{1[\bind]}=\true$ and $\phi_{2[\bind]}=\true$;
\item if $\phi = \phi_1 \lor \phi_2$ then $\phi_{1[\bind]}=\true$ or $\phi_{2[\bind]}=\true$.
\end{compactitem}

\medskip
In words, a guard is satisfied by evaluating it after assigning values to read and written variables, as specified by $\bind$. We can now define our \net{s}.

A \emph{state variable assignment}, abbreviated hereafter as \svassignment, is instead a function $\varState: \Vars \rightarrow \Universe\cup\set{\bot}$, which assigns values to each variable $\var\in \Vars$, with the restriction that $\varState(\typevar) \in \varDomain$. Note that this is different from variable assignments, which are defined over $(\VarsR\cup\VarsW)$. We can now define \net{s}.

\begin{definition}[Data Petri Net]
Let $\Vars$ be the set of process variables.
A \emph{Data Petri Net (\net)} $\anet=\DPN$ is a Petri net $(P,T,F)$ with additional components, used to describe the additional perspectives of the process model:
\begin{compactitem}
\item $\Vars$ is a finite set of process variables;
\item $\varDomAss$ is a function assigning a domain $\Domain$ to each $\var\in\Vars$;
\item $\varInit$ is the initial \svassignment;
\item $\reads: T \rightarrow \PS(\Vars)$ returns the set of variable \emph{read} by a transition;
\item $\writes: T \rightarrow \PS(\Vars)$ returns the set of variable \emph{written} by a transition;
\item $\guard: T \rightarrow \Expr(\Vars)$ returns a guard  associated with the transition, so that $\varr$ appears in $\guard(t)$ only if $\varname\in \reads(t)$, and $\varw$ appears in $\guard(t)$ only if $\varname\in \writes(t)$, for every $t$. 
\end{compactitem}
\end{definition}

\subsection{Execution Semantics}
\label{sec:semantics}
By considering the usual semantics for the underlying Petri net together with the guards associated to each of its transitions, we define the resulting execution semantics for \net{s}. 
First,
let $\anet$ as above be a \net. Then the set of possible states of $\anet$ is formed by all pairs $(M,\varState)$ where:
\begin{compactitem}
	\item $M \in \MS(P)$\footnote{$\MS(X)$ indicates the set of all multisets of elements of $X$}, i.e., is the marking of the Petri net $(P,T,F)$, and
	\item $\varState$ is a \svassignment, defined as in the previous section.
	\end{compactitem}

In any state, zero or more transitions of a \net may be able to fire. Firing a transition updates the marking,
reads the variables specified in $\reads(t)$ and selects a new, suitable value for those  in $\writes(t)$.
We model this through a variable assignment $\bind$ for the transition, which assigns a value to all and only those variables that are  read or written.
A pair $(t,\bind)$ is called \emph{transition firing}. 

\begin{definition}[Legal transition firing]
A \net $\anet=\DPN$ evolves from state $(M,\varState)$ to state $(M',\varState')$ via the transition firing $(t,\bind)$ with $\guard(t)=\phi$ iff:
\begin{compactitem}
\item $\bind(\varr) = \alpha(\var)$ if $\varname\in\reads(t)$: $\bind$ assigns values as $\alpha$ for read variables;
\item the new \svassignment is s.t. $\varState'(\var) = \left\{\begin{array}{lr}
\varState(\var) & \text{if} \; \varname\not\in \writes(t),\\
\bind(\varw) & \text{otherwise};\\
\end{array}\right.$\\
namely the new \svassignment $\varState'$ is as $\varState$ but updated as per $\bind$;
\item $\bind$ is \emph{valid}, namely $\phi_{[\bind]}=\true$: the guard is satisfied under $\bind$;
\item each input place of $t$ contains at least one token: $(M(p) > 0)$ for any $p \in P. (p,t) \in F$.
\item the new marking is calculated as usual, namely $M \goto{t} M'$.
\end{compactitem}
\end{definition}

We denote this by writing $(M,\varState) \goto{t,\bind} (M',\varState')$. We extend this to sequences $\sigma = \tup{(t^1,\bind^1),\ldots, (t^n,\bind^n)}$ of $n$ legal transition firings, called \emph{traces}, an denote the corresponding \emph{run} by
$(M^0,\varState^0) \goto{t^1,\bind^1} (M^1,\varState^1) \goto{t^2,\bind^2} \ldots \goto{t^n,\bind^n} (M^n,\varState^n)$ or equivalently by $(M^0,\varState^0)\goto{\sigma}(M^n,\varState^n)$.
By restricting to the initial marking $M_I$ of a \net $\anet$ together with the initial variable assignment $\varState_I$, we define the process traces of $\anet$ as the set of sequences $\sigma$ as above, of any length, such that $(M_I,\varState_I) \goto{\sigma} (M,\varState)$ for some $M\in \MS(P)$ and $\alpha$, and the \emph{trace set} of $\anet$ as the set of process traces $\sigma$ such that $(M_I,\varState_I) \goto{\sigma} (M_F,\varState)$ for some $\alpha$, where $M_F$ is the final marking of $\anet$.



\newcommand{\cguard}[1]{\ensuremath{\left[#1\right]}}
\newcommand{\cwrite}[1]{\textbf{write}:\ensuremath{#1}}

\newcommand{\cval}[1]{\mathtt{#1}}

\newcommand{\varok}{\mathit{ok}}
\newcommand{\varamount}{\mathit{amount}}

\newcommand{\readvar}[1]{#1^r}
\newcommand{\writevar}[1]{#1^w}

\tikzstyle{place} = [
  circle,
  very thick,
  draw=black,
  fill=white,
  minimum size=8mm,
  font=\fontsize{9}{144}\selectfont
]

\tikzstyle{transition} = [
  rectangle,
  very thick,
  draw=black,
  fill=white,
  minimum width=2cm,
  minimum height=1cm,
]

\tikzstyle{htransition} = [
  transition,
  fill=white,
  minimum width=8mm,
  minimum height=8mm,
]

\begin{figure}[t]
\centering
\resizebox{\textwidth}{!}{
\begin{tikzpicture}
  [ ->,
    >=stealth',
    auto,
    very thick,
    font=\sffamily,
    node distance=5mm
  ]

  \node[place, label={$i$}] (pcr) {};

  \node[transition, right= of pcr] (cr) {
    \begin{tabular}{c}
      credit\\ request
    \end{tabular}
  };
  \node[anchor=south] (gcr) at (cr.north){
    \cguard{\writevar{\varamount} \geq 0}
  };

  \node[place, right= of cr, label={$p_1$}] (pv) {};

  \node[transition, right=1.5cm of pv] (v) {
    \begin{tabular}{c}
      verify\\
    \end{tabular}
  };
  \node[anchor=south] (gv) at (v.north){
    \cguard{\writevar{\varok}}
  };

  \node[transition, below= (1.4cm) of pv] (r) {
    \begin{tabular}{c}
      renegotiate\\ request
    \end{tabular}
  };
  \node[anchor=south] (gr) at (r.north){
    \cguard{
      \begin{array}{@{}l@{}}
        \readvar{\varamount} > \cval{15000}
        \\
        {}\land \readvar{\varok} == \false
      \end{array}
    }
  };

  \node[transition, right= 1.8cm of r] (skipa) {
    \begin{tabular}{c}
      skip\\ assessment
    \end{tabular}
  };
  \node[anchor=south] (gskipa) at (skipa.north) {
    \cguard{
         \begin{array}{@{}l@{}}
        \readvar{\varok}==\false
        \\~
      \end{array}
    }
  };

  \node[place, right= 1.7cm of v, label={$p_2$}] (pa) {};

  \node[transition, right= 1.6cm of skipa] (simplea) {
    \begin{tabular}{c}
      simple\\ assessment
    \end{tabular}
  };
  \node[anchor=south] (gsimplea) at (simplea.north) {
    \cguard{
      \begin{array}{@{}l@{}}
        \readvar{\varok}==\true \land \writevar{\varok}\\
        {}\land \readvar{\varamount} < \cval{5000}
      \end{array}
    }
  };

  \node[transition, right= 2.8cm of simplea] (advanceda) {
    \begin{tabular}{c}
      advanced\\ assessment
    \end{tabular}
  };
  \node[anchor=south] (gadvanceda) at (advanceda.north) {
    \cguard{
      \begin{array}{@{}l@{}}
        \readvar{\varok}==\true \land \writevar{\varok} \\
        {}\land \readvar{\varamount} \geq \cval{5000}
      \end{array}
    }
  };

  \node[place, below= 4mm of simplea, label=350:{$p_3$}] (rp) {};


  \node[htransition, below = 4mm of rp] (et1) { AND split};

  \node[place,
        right= (2.6cm) of et1,
        label=right:{$p_5$}] (op) {};

  \node[place,
        left= (4.5cm) of et1,
        label=left:{$p_4$}] (ip) {};

  \node[transition, below= (1.4cm) of op] (o) {
    \begin{tabular}{c}
      open\\ credit loan
    \end{tabular}
  };

  \node[anchor=south] (go) at (o.north){
    \cguard{\readvar{\varok} == \true}
  };

  \node[transition, below= 1.4cm of ip] (ivip) {
    \begin{tabular}{c}
      inform acceptance\\ customer VIP
    \end{tabular}
  };

    \node[transition, left= of ivip] (inormal) {
    \begin{tabular}{c}
      inform acceptance\\ customer normal
    \end{tabular}
  };
  \node[anchor=south] (ginormal) at (inormal.north){
    \cguard{
      \begin{array}{@{}l@{}}
        \readvar{\varok}==\true\\
        {}\land \readvar{\varamount} < \cval{10000}
      \end{array}
    }

  };

  \node[transition, right= of ivip] (irvip) {
    \begin{tabular}{c}
      inform rejection \\ customer VIP
    \end{tabular}
  };

  \node[anchor=south] (givip) at (ivip.north){
    \cguard{
      \begin{array}{@{}l@{}}
        \readvar{\varok}==\true\\
        {}\land \readvar{\varamount} \geq \cval{10000}
      \end{array}
    }
  };

  \node[anchor=south] (girvip) at (irvip.north){
    \cguard{
      \begin{array}{@{}l@{}}
        \readvar{\varok}==\false\\
        {}\land \readvar{\varamount} \geq \cval{10000}
      \end{array}
    }
  };

  \node[place,
        below= (2.8cm) of op,
        label=right:{$p_7$}] (ps1) {};

  \node[place,
        below= (2.8cm) of ip,
        label=left:{$p_6$}] (ps2) {};

  \node[htransition, below= (2.8cm) of et1] (et2) { AND join };

  \node[place,
        above= of et2,
        label={$o$}] (pout) {};

  \draw[->] (pcr) edge (cr);
  \draw[->] (cr) edge (pv);
  \draw[->] (pv) edge (v);
  \draw[->] (v) edge (pa);
  \draw[->] (pa) edge (gskipa); 
  \draw[->] (pa) edge (gsimplea);
  \draw[->] (pa) edge (gadvanceda);
  \draw[->] (skipa) edge (rp);
  \draw[->] (simplea) edge (rp);
  \draw[->] (advanceda) edge (rp);
  \draw[->, rounded corners=10pt] (rp) -| (r);
  \draw[->] (gr) edge (pv);
  \draw[->] (rp) edge (et1);
  \draw[->] (et1) edge (op);
  \draw[->] (et1) edge (ip);
  \draw[->] (op) edge (go);
  \draw[->] (ip) edge (ginormal);
  \draw[->] (ip) edge (givip);
  \draw[->] (o) edge (ps1);
  \draw[->] (inormal) edge (ps2);
  \draw[->] (ivip) edge (ps2);
  \draw[->] (irvip) edge (ps2);
  \draw[->] (ip) edge (girvip);
  \draw[->] (ps1) edge (et2);
  \draw[->] (ps2) edge (et2);
  \draw[->] (et2) edge (pout);

\end{tikzpicture}
}
\vspace{-.1in}
\caption{An example of \net, which will be used as working example throughout the paper. Writing and reading operations are omitted for readability. In this example, writing operations exist everytime guards mention $\writevar{\varok}$ or $\writevar{\varamount}$. Terms $\writevar{\varok}$ in transitions \textsf{verify}, \textsf{simple assessment} and \textsf{advanced assessment} are only intended to explicitly indicate that the variable is simply written but can take on either \true or \false.}
\label{fig:net}
\vspace{-.12in}
\end{figure}
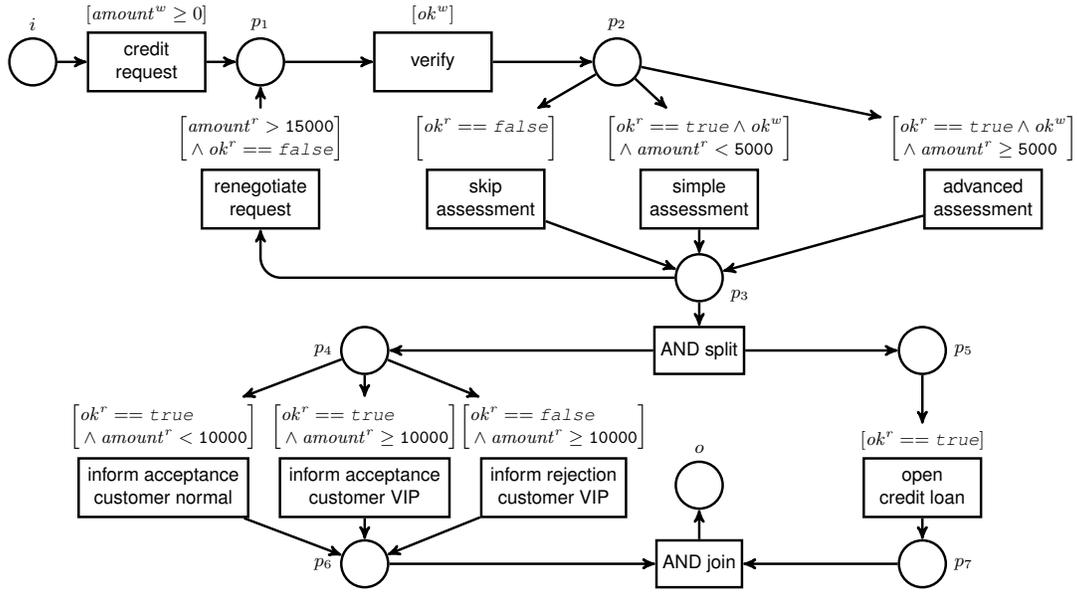

\begin{example}
Figure~\ref{fig:net} shows a \net representing a process for managing credit requests and corresponding loans. The \net employs two case variables, $\varok$ and $\varamount$, respectively used to capture whether the credit request is accepted or not, and what is the requested amount. The process starts by acquiring the amount of the credit request (thus writing $amount$), which must be positive. Then a verification step is performed, determining whether to accept of reject the request (thus writing $ok$). In the rejection case, a new verification may be performed provided that the requested amount exceeds $15000$ euros (\textsf{\small{skip assessment}} followed by \textsf{\small{renegotiate request}}). In the acceptance case, depending on the requested amount, a simple or advanced assessment is performed. The second phase of the process then deals, concurrently, with the opening of a loan (which can only be executed if the request is accepted), and with a communication sent to the customer, which depends again on the combination of data hold by the case variables.
In Figure~\ref{fig:runs} we compactly represent some run fragments. As shown in the figure, the number of legal runs is infinite (e.g., the number of possible values for $\varamount$ is infinite) and also their length may be unbounded (due to cycles in the process).
\label{ex:example1}
\end{example}

We are interested  in characterising properties of \net{s}. For this reason, is it useful to compare these nets by looking at their behaviour, i.e. their trace set.

This is achieved in two steps. We first define the notion of \emph{trace-equivalence}, which will also be helpful for proving our results.

\begin{definition}[Trace-equivalence between \net{s}]
Given two runs $(M_{I1},\varState^0_1) \goto{\sigma_1} (M_1^n,\varState_1^n)$ and $(M_{I2},\varState_2^0) \goto{\sigma_2} (M_2^n,\varState_2^n)$ of two \net{s} $\anet_1$ and $\anet_2$, respectively, these runs are trace-equivalent iff $M_{I1}=M_{I2}$ and for any $i\in [1,n]$ we have that $t^i_1=t^i_2$, namely the transitions are the same.
\end{definition}

Similarly, two \net{s} $\anet_1$ and $\anet_2$ are \emph{trace-equivalent} iff $(a)$ for every legal run $\tau_1$ of $\anet_1$ there exists a trace-equivalent run $\tau_2$ of $\anet_2$ and $(b)$ vice-versa.

Note that for any \net, given a state $(M,\varState)$ and a legal transition firing $(t,\bind)$ from that state, there exists exactly one successor state $(M',\varState')$ such that $(M,\varState)\goto{t,\bind}(M',\varState')$, namely the \net is transition-deterministic (for a given binding).
As a consequence, two runs that are trace-equivalent also traverse the same markings, namely $M^i_1=M^i_2$, $i\in[1,n]$.

\begin{figure}[t]
\centering
\resizebox{\textwidth}{!}{

\begin{tikzpicture}[>=stealth',
    auto,
    thick, node distance = 6cm]

\node[state, label={90:$i, \begin{pmatrix} \varok& \rightarrow & \bot\\ \varamount & \rightarrow & \bot \end{pmatrix}$}] (0) {};
\node[state, label={90:$p_1, \begin{pmatrix} \varok& \rightarrow & \bot\\ \varamount & \rightarrow & 1400 \end{pmatrix}$}, right of=0, xshift=-.4cm] (1) {};
\node[state, label={90:$p_2, \begin{pmatrix} \varok& \rightarrow & \true\\ \varamount & \rightarrow & 1400 \end{pmatrix}$}, right of=1, xshift=-1.5cm] (2) {};
\node[state, label={90:$p_3, \begin{pmatrix} \varok& \rightarrow & \false\\ \varamount & \rightarrow & 1400 \end{pmatrix}$}, right of=2, xshift=.4cm] (3) {};

\node[state, label={90:$p_2, \begin{pmatrix} \varok& \rightarrow & \false\\ \varamount & \rightarrow & 1400 \end{pmatrix}$}, below of=2, yshift=3.5cm] (4) {};

\node[state, label={90:$p_1, \begin{pmatrix} \varok& \rightarrow & \bot\\ \varamount & \rightarrow & 3600 \end{pmatrix}$}, below of=1, yshift=1cm] (5) {};
\node[state, label={90:$p_2, \begin{pmatrix} \varok& \rightarrow & \true\\ \varamount & \rightarrow & 3600 \end{pmatrix}$}, right of=5, xshift=-1.5cm] (6) {};

\node[state, label={90:$p_2, \begin{pmatrix} \varok& \rightarrow & \false\\ \varamount & \rightarrow & 3600 \end{pmatrix}$}, below of=6, yshift=3.5cm] (7) {};

\node[state, label={90:$p_1, (\cdots) $}, label={270:$\cdots$}, below right of=0, shift={(-2cm,1.7cm)}] (n0) {};
\node[right of=n0, xshift=-4.7cm] (n00) {};
\node[right of=n0, yshift=-1cm, xshift=-4.7cm] (n01) {};
\node[right of=3, xshift=-4.7cm] (n3) {};
\node[right of=4, xshift=-4.7cm] (n4) {};
\node[right of=6, xshift=-4.7cm] (n6) {};
\node[right of=7, xshift=-4.7cm] (n7) {};
\node[state,above of=3,yshift=-3.7cm] (n8) {};
\node[right of=n8, xshift=-4.7cm] (n9) {};

\draw[->] (0) edge  node[below] {$(\textsf{credit req.}, (\varamount^w=1400))$}  (1);
\draw[->] (1) edge  node[below] {$(\textsf{verify}, (\varok^w \rightarrow \true))$}  (2);
\draw[->] (2) edge  node[below] {$(\textsf{simple assess.}, \begin{pmatrix} \varok^r & \rightarrow & \true\\ \varamount^r & \rightarrow & 1400 \\
\varok^w & \rightarrow & \false\end{pmatrix})$}  (3);

\draw[->,rounded corners=10pt] (1) |- node[below, pos=.75] {$(\textsf{verify}, (\varok^w \rightarrow \false))$}  (4);

\draw[->,rounded corners=10pt] (0) |- node[below, pos=.75] {$(\textsf{credit req.}, (\varamount^w \rightarrow 3600))$}  (5);

\draw[->] (5) edge  node[below] {$(\textsf{verify}, (\varok^w \rightarrow \true))$}  (6);
\draw[->,rounded corners=10pt] (5) |- node[below, pos=.75] {$(\textsf{verify}, (\varok^w \rightarrow \false))$}  (7);

\draw[->, dashed, out=290, in=180] (0) edge  node[right, pos=.2, xshift=.2cm] {$(\textsf{credit req.}, (\cdots))$}  (n0);
\draw[->, dashed] (n0) edge  node {} (n00);
\draw[->, dashed] (n0) edge  node {} (n01);
\draw[->, dashed] (3) edge  node {} (n3);
\draw[->, dashed] (4) edge  node {} (n4);
\draw[->, dashed] (6) edge  node {} (n6);
\draw[->, dashed] (7) edge  node {} (n7);
\draw[->, dashed] (2) edge  node {} (n8);
\draw[->, dashed] (n8) edge  node {} (n9);

\end{tikzpicture}
}
\caption{Example of run fragments of the working example, assuming an initial \svassignment $\varInit$ in which both case variables are not set. Arcs are labelled with legal transition firings of the form $(t,\bind)$.}
\label{fig:runs}
\end{figure}

\subsection{Data-aware Soundness}
\label{sec:soundness}
We now lift the standard notion of soundness \cite{Aalst1998workflow} to the case of \net{s}.
This requires to quantify not only over the markings of the net, but also on the assignments of its case variables, thus making soundness \emph{data-aware} (we use `data-aware' to distinguish our notion from the one of decision-aware soundness in the literature -- see Section\ref{sec:data-vs-decision}).
In what follows, we write $(M,\varState) \goto{*} (M',\varState')$ to implicitly quantify \emph{existentially} on sequences $\sigma$.

\begin{definition}[Data-aware soundness] A \net is \emph{data-aware sound} iff the following properties hold:
\begin{compactitem}
\item[P1:] $\forall (M, \varState)$. $((M_I,\varInit) \goto{*} (M,\varState)$ $\Rightarrow$ $\exists \varState'$. $(M,\varState) \goto{*} (M_F,\varState'))$
\item[P2:] $\forall (M, \varState)$. $(M_I,\varInit) \goto{*} (M,\varState) \land M\geq M_F$ $\Rightarrow$ ($M=M_F$)
\item[P3:] $\forall t\in T$. $\exists M_1, M_2, \varState_1,\varState_2,\bind$. $(M_I,\varInit) \goto{*} (M_1,\varState_1) \goto{t,\bind} (M_2,\varState_2)$
\end{compactitem}
\label{def:properties}
\end{definition}

The first condition checks the reachability of the output state, that is, whether it is \emph{always} possible to reach the final marking of $\anet$ by suitably choosing a continuation of the current run (i.e., transitions and variable assignments). The second condition captures that the output state is reached in a clean way, i.e., that $\anet$ cannot reach the final marking while \emph{in addition} having other tokens in other places. The third condition verifies the absence of dead transitions, where a transition is considered dead if there is no way of assigning the case variables so as to enable it.

\begin{example}
\label{ex:unsound}
Consider again the \net in Figure~\ref{fig:net}. Such a \net is unsound for a number of reasons, related to the concurrent section in the second phase of the process. Suppose that the verification step assigns $\varok$ to $\false$. Once the execution assigns a token to $p_3$, and the following AND-split transition is fired, two tokens are produced, respectively placing them in $p_4$ and $p_5$. Since the guard of \textsf{open credit loan} is false, token $p_5$ cannot be consumed, and thus it is not possible to properly complete the execution. In addition, if the requested amount is less than $\cval{10000}$, the same occurs also for the token placed in $p_4$.
\end{example}


\newcommand{\domset}{\set{\Domain_{\Integers},\Domain_{\Reals},\Domain_{\mathit{bool}},\Domain_{\mathit{string}}}}
\newcommand{\doms}{\mathfrak{D}}
\newcommand{\dnet}[1]{\ensuremath{#1}-\net\xspace}
\newcommand{\dnets}[1]{\ensuremath{#1}-\net{s}\xspace}
\newcommand{\varat}{\mathit{atype}}
\newcommand{\anorm}{\cval{simple}}
\newcommand{\anone}{\cval{none}}
\newcommand{\aadv}{\cval{advanced}}

\definecolor{incolor}{RGB}{210,220,230}
\definecolor{outcolor}{RGB}{235,215,215}

\newcommand{\dmntext}[1]{\phantom{$\mathtt{[}$}#1\phantom{$\mathtt{]}$}}
\tikzset{
table/.style={
  matrix of nodes,
  row sep=-\pgflinewidth,
  column sep=-\pgflinewidth,
  nodes={
    rectangle,
    draw=black,
    minimum width=.7cm,
    minimum height=5mm,
    align=center },
  text depth=0.25ex,
  text height=1ex,
  nodes in empty cells
  },
  dmn/.style={
    matrix of nodes,
    row sep=-\pgflinewidth,
    column sep=-\pgflinewidth,
    nodes={
      rectangle,
      draw=black,
      text width=12mm,
      minimum height=5mm,
      align=center },
    nodes in empty cells,
  },
  dmnhit/.style={
    rectangle,
    draw,
    minimum height=10.55mm,
    minimum width=7.1mm,
    xshift=1.3mm
  },
  dmnrulen/.style={
    matrix of nodes,
    row sep=-\pgflinewidth,
    column sep=-\pgflinewidth,
    nodes={
      rectangle,
      draw=black,
      text width=5mm,
      minimum height=5mm,
      align=center },
    nodes in empty cells,
  },
}

\section{Modeling with \net{s}}
\label{sec:dmn}

From now on, we always consider \net{s} working over the notable set $\doms = \domset$ of domains introduced at the beginning of Section~\ref{sec:syntax-semantics}.
We show that this class of \net{s} is expressive enough to directly incorporate in the model decisions expressed using the OMG standard DMN S-FEEL language \cite{DMN.2011,CDL16}.
Specifically, we first discuss how \net{s} can be enriched with such decision constructs, arguing that the so-obtained extended model captures those studied in the literature \cite{Batoulis2017,BaHW17}. We then show that such an extension is syntactic sugar, as it can be encoded back into standard \net{s}. This implies that the results presented in this paper can be seamlessly used to formalize the interesting decision-aware process models studied in \cite{Batoulis2017,BaHW17}, and check their soundness considering the different variants of soundness as defined in \cite{BaHW17}, as we will show in Section~\ref{sec:data-vs-decision}.

\subsection{\net{s} with DMN Decisions}
\label{sec:dmn-nets}
The integration of DMN decision with models capturing the control flow of a process, such as workflow nets, has been recently studied in \cite{Batoulis2017,BaHW17}. As argued in \cite{Batoulis2017,BaHW17}, using Petri nets to capture the process control flow does not incur in loss of generality: the integration can be in fact conceptually captured at a higher level of abstraction, such as that of the combination of DMN with BPMN, then applying standard control-flow translation mechanisms \cite{DDO08} to encode the control flow of the input BPMN model into a corresponding Petri net.

The standard way of incorporating a DMN decision into a BPMN process is to introduce a business rule task in the process. This task, in turn, is linked to the DMN decision. Whenever the business rule task is reached during the execution of a process instance, the inputs of the decision are bound to specific values, and the corresponding output result is calculated and incorporated into the state of the process instance for further use. This also corresponds to the notion of \emph{decision fragment} in \cite{BaHW17}.
In the context of \net{s}, the natural incorporation of a DMN decision consequently amounts to introduce a special \emph{decision transition}  that is linked to a DMN decision. Since \net{s} are natively equipped with case data, we assume that the inputs and outputs of the decision coincide with (some of) the case variables of the \net.

\begin{example}
\label{ex:dmn-net}
Consider a variant of the \net shown in Figure~\ref{fig:net}, where we want to explicitly track the type of assessment that must be conducted on a given credit request, from place $p_2$. Therefore, we can transform the three transitions from $p_2$ into the rows of a decision table, and use an additional case variable $\varat$, of type string, as output of the table and consequently in the conditions of the branches of the split-gateway, as shown in Figure~\ref{fig:dmn-net}. Such a variable can be assigned to one among the strings $\anone$, $\anorm$, $\aadv$, respectively indicating no assessment, normal assessment, and advanced assessment.
To do so, we extract the decision logic distributed over the outgoing arcs from place $p_2$ in Figure~\ref{fig:net}, and combine the conditions therein into a single DMN decision, which indicates how the value $\varat$ is computed depending on the values of the two input variables $\varok$ and $\varamount$. Then, we attach this DMN decision to a dedicated decision transition, which is in turn inserted in the net between the verification and assessment steps. Finally, we update the three assessment transitions, associating each of them to its corresponding value for $\varat$. The resulting decision fragment is shown in Figure~\ref{fig:dmn-net}.

\begin{figure}[t]
\centering
\resizebox{.7\textwidth}{!}{
\begin{tikzpicture}[node distance =0pt and 0.5cm]

\begin{scope}[]

\matrix[dmn, xshift=-.2cm] 
  (in) {
  |[fill=incolor]|
  \dmntext{$\varamount$} 
  & 
  |[fill=incolor]|
  \dmntext{$\varok$} 
  \\
  \dmntext{$\geq 0$} &
  \dmntext{}
  \\
};

\matrix[dmn,anchor=north west, xshift=-2mm, text width=2cm] 
  (out) 
  at (in.north east) {
  |[fill=outcolor,text width=27mm]|
  \dmntext{$\varat$}
  \\
  |[text width=27mm]|\dmntext{\scriptsize$\anone,\anorm,\aadv$}
  \\
};

\node[dmnhit,anchor=east]
  (hit)
  at (in.west)
  {\textbf{U}};
  
\node[rectangle,
      draw,
      anchor=south west,
      minimum height=5mm,
      minimum width=45mm,
      yshift=-\pgflinewidth]
  (title)
  at (hit.north west)
  {\textbf{Determine Type of Assessment}};
  
\matrix[dmn,anchor=north west,yshift=1.9mm]
  (rules)
  at (in.south west) {
    \dmntext{--}
    &
    \dmntext{$\false$}
     \\
    \dmntext{\scriptsize$< \cval{5000}$}
    &
    \dmntext{$\true$}
    \\
     \dmntext{\scriptsize$\geq \cval{5000}$}
    &
    \dmntext{$\true$}
    \\
  };
  
  \matrix[dmnrulen,
        anchor=north east,
        xshift=2.4mm] 
  (rulenum)
  at (rules.north west) {
    \dmntext{1}\\
    \dmntext{2}\\
    \dmntext{3}\\
};
  
  \matrix[dmn,
        anchor=north west,
        yshift=1.9mm]
  (ruleconclusions)
  at (out.south west) {
  |[text width=27mm]|\dmntext{$\anone$} \\
  |[text width=27mm]|\dmntext{$\anorm$} \\
  |[text width=27mm]|\dmntext{$\aadv$}  \\
};

%
\node[place, 
      label={[xshift=2mm,yshift=4mm]left:$p_2$}] (p) at (1.5,-3.2) {};
\node[htransition, right= 7.5mm of p] (t) {};    
\node[place, 
      right= 7.5mm of t,
      label={[xshift=2mm,yshift=4mm]left:$p_2'$}] (p2) {};

   \node[font=\sffamily,transition, right = 7.5mm of p2] (advanceda) {
    \begin{tabular}{c}
      advanced\\ assessment
    \end{tabular}
  };
  \node[anchor=south] (gadvanceda) at (advanceda.north) {
    \cguard{\varat == \aadv}
  };

  \node[font=\sffamily,transition, above = 6mm of advanceda] (normala) {
    \begin{tabular}{c}
      simple\\ assessment
    \end{tabular}
  };
  \node[anchor=south] (gnormala) at (normala.north) {
    \cguard{\varat == \anorm}
  };
    
  \node[font=\sffamily,transition, above = 6mm of normala] (skipa) {
    \begin{tabular}{c}
      skip\\ assessment
    \end{tabular}
  };
  \node[anchor=south] (gskipa) at (skipa.north) {
    \cguard{\varat == \anone}
  };

  
  \node[place,
        right= 3.5 cm of p2,
        label={[xshift=2mm,yshift=4mm]left:$p_3$}] (out) {};

  \node[
    draw,
    densely dotted,
    ultra thick,
    inner sep= 1.5mm,
    rounded corners=10pt,
    xshift=-.8mm,
    yshift=.8mm,
    fit = (ruleconclusions.south east) (title.north west)
    ] (dmncover) {};

  \draw[->, very thick] (p) edge (t);
  \draw[->, very thick] (t) edge (p2);
  \draw[->, very thick, rounded corners=10pt] (p2) |- (skipa);
  \draw[->, very thick, out=70, in=180] (p2) edge (normala);
  \draw[->, very thick, rounded corners=10pt] (p2) edge (advanceda);
  \draw[->, very thick, rounded corners=10pt] (skipa) -| (out);
 \draw[->, very thick, rounded corners=10pt,out=0,in=110] (normala) edge (out);
 \draw[->, very thick, rounded corners=10pt] (advanceda) edge (out);
 \draw[ultra thick,-,densely dotted] (t) edge (dmncover);

   \end{scope}

    \end{tikzpicture}
     
}
  \caption{Fragment of an extended \net equipped with a DMN decision transition, corresponding to the decision logic of the transitions outgoing from $p_2$ in Figure~\ref{fig:net}.}
  \label{fig:dmn-net}
\end{figure}
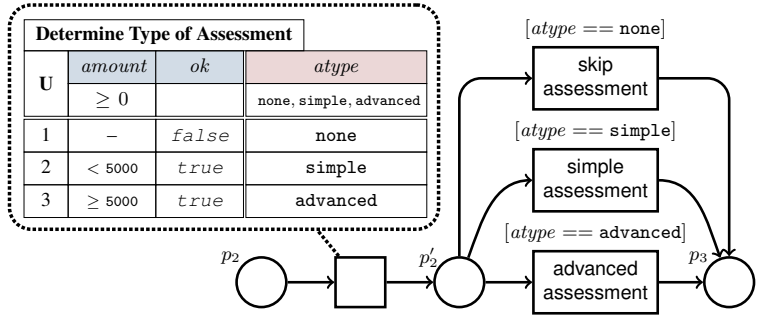

\end{example}


This extension of \net{s} with DMN-based decision transitions captures the decision-aware models recently studied in \cite{Batoulis2017,BaHW17}. On the one hand, we reconstruct the decision transitions defined there. On the other hand, we explicitly account for case variables and for (guarded) updates of their values, introducing a source of nondeterminism that depends on picking a new value for the updated variable among a possibly infinite set of potential values.

When considering BPMN as an input specification language, we produce a corresponding \net as follows:
\begin{compactitem}
\item As for the control flow, we apply the same translation from BPMN to Petri nets adopted in \cite{BaHW17}, namely the one described in \cite{DDO08};
\item For each data object name in the BPMN model, we introduce a  case variable with the same name. We only deal with data object collections whose (largest) size is known a-priori, so that a dedicated case variable is produced for each element of the collection.
\item Whenever a BPMN activity connects to a data object with name $n$, we ensure that the corresponding \net transition writes the variable mirroring that data object, i.e., we set its guard to be the formula $\cguard{\writevar{n}}$;
\item If the BPMN diagram predicates over the states of an object with name $n$, we introduce in the \net a ``state'' variable $n_{state}$ of type string, to keep track of the current state of $n$.
\item If a BPMN activity requires an object with name $n$ to be in a given state $\cval{s}$ prior execution, we guard the corresponding \net transition with condition $\cguard{\readvar{n_{state}} == \cval{s}}$.
\item If a BPMN activity updates an object with name $n$ to state $\cval{s'}$ upon completion, we guard the corresponding \net transition with condition $\cguard{\writevar{n_{state}} == \cval{s'}}$.
\end{compactitem}

\subsection{Encoding \net{s} with DMN Decisions to Normal \net{s}}
\label{sec:encoding}
We now show that the DMN S-FEEL extension proposed in Section~\ref{sec:dmn-nets} is actually syntactic sugar, in the sense that its induced decision logic can be mimicked by a normal \net. In what follows, we restrict the attention to decision tables with \emph{unique} hit policies, although other policies can be considered as well by introducing a case variable for each subset of possible outputs of the decision table. This however generates a combinatorial explosion.

We describe here the transformation intuitively, because a formal description would be too cumbersome.
Consider a \net $\anet_{\textit{DMN}}$ extended with DMN decision transitions.
Intuitively, we need to transform the application of each rule in the decision table, together with the successive branch in the split-gateway which covers it, into a simple transition $t$ with $\guard(t)$ encoding all the condition of the rule on both input variables (read variables) and output variables (written variables). In Figure~\ref{fig:dmn-net} we show an intuitive example. Notice that, whenever there exist more than one decision tasks in $\anet_{\textit{DMN}}$ that are possible from the same place, to correctly preserve the independence of these tasks (and that of their decision tables), we need to introduce \emph{internal} transitions.

\begin{figure}[b]
\centering
\resizebox{\textwidth}{!}{

\begin{tikzpicture}[node distance =0pt and 0.5cm]
\tikzstyle{transition}=[rectangle,thick,draw=black,  inner sep=6pt]

\begin{scope}[]

\matrix[dmn, xshift=-.2cm] 
  (in) {
  |[fill=incolor]|
  \dmntext{$class$} 
  & 
  |[fill=incolor]|
  \dmntext{$points$} 
  \\
  \dmntext{} &
  \dmntext{}
  \\
};

\matrix[dmn,anchor=north west, xshift=-2mm, text width=2cm] 
  (out) 
  at (in.north east) {
  |[fill=outcolor,text width=22mm]|
  \dmntext{$disc$}
  \\
  |[text width=22mm]|\dmntext{\scriptsize$\set{25\%, 35\%, 60\%}$}
  \\
};

\node[dmnhit,anchor=east]
  (hit)
  at (in.west)
  {\textbf{U}};
  
\node[rectangle,
      draw,
      anchor=south west,
      minimum height=5mm,
      minimum width=33mm,
      yshift=-\pgflinewidth]
  (title)
  at (hit.north west)
  {\textbf{Determine discount}};
  
\matrix[dmn,anchor=north west,yshift=1.9mm]
  (rules)
  at (in.south west) {
    \dmntext{\scriptsize$2$}
    &
    \dmntext{\scriptsize$-$}
     \\
    \dmntext{\scriptsize$1$}
    &
    \dmntext{\scriptsize$< \cval{100}$}
    \\
     \dmntext{\scriptsize$1$}
    &
    \dmntext{\scriptsize$\geq \cval{100}$}
    \\
  };
  
  \matrix[dmnrulen,
        anchor=north east,
        xshift=2.4mm] 
  (rulenum)
  at (rules.north west) {
    \dmntext{1}\\
    \dmntext{2}\\
    \dmntext{3}\\
};
  
  \matrix[dmn,
        anchor=north west,
        yshift=1.9mm]
  (ruleconclusions)
  at (out.south west) {
  |[text width=22mm]|\dmntext{\scriptsize$25\%$} \\
  |[text width=22mm]|\dmntext{\scriptsize$35\%$} \\
  |[text width=22mm]|\dmntext{\scriptsize$60\%$}  \\
};

%
\node[place, 
      label={}] (p) at (1.1,-3.2) {$p_A$};
\node[htransition, right= 7.5mm of p] (t) {$t_d$};    
\node[place, 
      right= 7.5mm of t,
      label={}] (p2) {$p_B$};

   \node[transition, right = 7.5mm of p2] (advanceda) { $t_3$  };
  \node[anchor=south] (gadvanceda) at (advanceda.north) {
    \cguard{disc == 60\%}
  };

  \node[transition, above = 6mm of advanceda] (normala) { $t_2$  };
  \node[anchor=south] (gnormala) at (normala.north) {
    \cguard{disc == 35\%}
  };
    
  \node[font=\sffamily,transition, above = 6mm of normala] (skipa) { $t_1$ };
  \node[anchor=south] (gskipa) at (skipa.north) {
    \cguard{disc == 25\%}
  };
  
    \node[font=\sffamily,transition, below = 10mm of gadvanceda, label={below:$[age>65]$}] (bbbb) { $t_4$ };


  \node[
    draw,
    densely dotted,
    ultra thick,
    inner sep= 1.5mm,
    rounded corners=10pt,
    xshift=-.8mm,
    yshift=.8mm,
    fit = (ruleconclusions.south east) (title.north west)
    ] (dmncover) {};

  \draw[->, very thick] (p) edge (t);
  \draw[->, very thick] (t) edge (p2);
  \draw[->, very thick, rounded corners=10pt] (p2) |- (skipa);
  \draw[->, very thick, out=70, in=180] (p2) edge (normala);
  \draw[->, very thick, rounded corners=10pt] (p2) edge (advanceda);
 \draw[ultra thick,-,densely dotted] (t) edge (dmncover);
  \draw[->, very thick, rounded corners=10pt] (p)  |- (bbbb);

   \end{scope}

\begin{scope}[shift={(6cm,0cm)}, node distance =0pt and 0.5cm]

\node[place] (p) at (1.5,-3.2) {$p_A$};
\node[htransition,  right= 7.5mm of p, fill=black, text=white] (t) {\textbf{$t_d$}};    
\node[place, right= 7.5mm of t, label={}] (p2) {$p_B$};

   \node[transition, right = 7.5mm of p2] (advanceda) { $t_3$  };
  \node[right of=advanceda] (gadvanceda) [text width=4cm,align=center, xshift=26mm]  {
    $class^r = 1 \land points^r\geq 100$\\$\land~disc^w = 60\%$
      };

  \node[transition, above = 6mm of advanceda] (normala) { $t_2$  };
  \node[right of=normala] (gnormala) [text width=4cm,align=center, xshift=26mm]  {
    $class^r = 1 \land points^r< 100$\\$\land~disc^w = 35\%$
  };
    
  \node[font=\sffamily,transition, above = 6mm of normala] (skipa) { $t_1$ };
  \node[right of=skipa] (gskipa) [text width=4cm,align=center, xshift=26mm]  {
    $class^r = 2 \land disc^w = 25\%$
  };
  
    \node[font=\sffamily,transition, below = 5mm of advanceda, label={right:$~~~age^r>65$}] (bbbb) { $t_4$ };


  \node[
    draw,
    densely dotted,
    ultra thick,
    inner sep= 1.5mm,
    rounded corners=10pt,
    xshift=-.8mm,
    yshift=.8mm,
    fit = (ruleconclusions.south east) (title.north west)
    ] (dmncover) {};

  \draw[->, very thick] (p) edge (t);
  \draw[->, very thick] (t) edge (p2);
  \draw[->, very thick, rounded corners=10pt] (p2) |- (skipa);
  \draw[->, very thick, out=70, in=180] (p2) edge (normala);
  \draw[->, very thick, rounded corners=10pt] (p2) edge (advanceda);
  \draw[->, very thick, rounded corners=10pt] (p)  |- (bbbb);

   \end{scope}

  \end{tikzpicture}
}
  \caption{Fragment of an extended \net equipped with a DMN decision transition translated into a fragment of a regular \net. The same translation can be applied to obtain a \net from process models equipped with DMN decisions as those in \cite{BaHW17}, which allow decisions modelled as \emph{decision fragments}. Note how an internal transition was used to keep the decision between transitions $t_1,t_2,t_3$ separate from $t_4$ (the same can be done when instead of $t_4$ we have another DMN decision). This is consistent with the intended semantics of properties $P1$-$P3$ in Definition~\ref{def:properties}, which requires that the process can be completed from every possible reachable state. Namely, for the \net (right) to be (data-aware) sound, there must be a way to complete the process from any reachable state $(M,\varState)$, including those in which $M$ is so that there are tokens in $p_A$ or $p_B$. This must be guaranteed irrespective of the value of the variables (\svassignment $\varState$). Hence, if the \net (right) is sound, this also implies that the decision table (left) is conditionally complete and conditionally output covered, as defined in \cite{BaHW17}. This will be discussed further in Section~\ref{sec:data-vs-decision}. }
  \label{fig:dmn-net}
\end{figure}
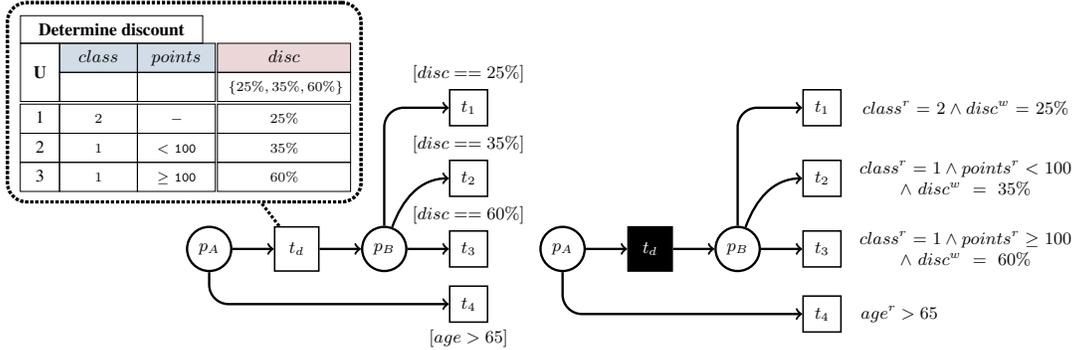



\section{Soundness Verification}
\label{sec:verification}



Coloured Petri Nets (\cpn{s}) are an extension to Data Petri Nets that have a better support for time and resource~\cite{JensenCPN:2009}.
Furthermore, \cpn{s} can be simulated through CPN Tools~\cite{RatzerCPNTools}, which makes it possible to build on existing techniques to compute soundness.
Differently from Data Petri Nets where variables are global, \cpn{s} encode the data aspects in the tokens, allowing tokens to have a data value, called \emph{color}, attached to them. Each place in a \cpn{s} usually contain tokens of one type, and this type is called \emph{color set} of the place.

Figure~\ref{fig:cpnExplanation} illustrates a \cpn{s}. Differently from Petri nets and \net{s}, tokens are associated with values (e.g.\ \emph{low} or \emph{high} in our example). When a transition fires, e.g.\ \emph{check\_low}, tries to consume one of the tokens, e.g.\ the token with value \emph{high}, and assign the token value to the variable on the arc, i.e.\ variable $q$ takes on value \emph{high}. This variable assignment  (a.k.a.\ binding) is valid if it does not violate the possible guard. In the example, the guard states the $q$ must be given value \emph{low}. This means that tokens with value \emph{high} cannot be consumed by transition \emph{check\_low}. Conversely, tokens with value \emph{high} can be consumed by transition \emph{check\_high}. All places in this example of \cpn are allowed to contain tokens associated with an enumerated type $\{high, low\}$, with the latter being the so-called color set associated with every place of this \cpn.

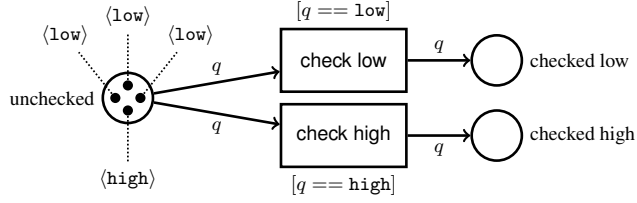
\begin{figure}[t]
\centering
\resizebox{.6\textwidth}{!}{
  \begin{tikzpicture}
  
   \begin{scope}[shift={(0cm,0cm)}]
  
    \node[place,label=left:unchecked] (u) at (0,0) {};
    \node[token] (token1) at (-.2,0) {\phantom{a}};
    \node[token] (token2) at (0,.2) {\phantom{a}};
    \node[token] (token3) at (.2,0) {\phantom{a}};
    \node[token] (token4) at (0,-.2) {\phantom{a}};
    
    \node[transition, 
          right= 2cm of u,
          font=\sffamily,
          yshift=6mm] (tl) {check low};
    \node[anchor=south] at (tl.north) {
      \cguard{q == \cval{low}}
    };
    
    \node[transition, 
          right= 2cm of u,
          font=\sffamily,
          yshift=-6mm] (th) {check high};
    \node[anchor=north] at (th.south) {
      \cguard{q == \cval{high}}
    };
    
    \node[place,
          label= right:checked low,
          right= of tl] (pl) {};

    \node[place,
          label= right:checked high,
          right= of th] (ph) {};
    
    \draw[->,very thick] (u) edge node[anchor=south] {$q$} (tl); 
    \draw[->,very thick] (u) edge node[anchor=north] {$q$} (th);
    \draw[->,very thick] (tl) edge node[anchor=south] {$q$} (pl); 
    \draw[->,very thick] (th) edge node[anchor=north] {$q$} (ph);

    \node (token1c) at (-1,1) {$\tup{\cval{low}}$};
    \node (token2c) at (0,1.3) {$\tup{\cval{low}}$};
    \node (token3c) at (1,1) {$\tup{\cval{low}}$};
    \node (token4c) at (0,-1.3) {$\tup{\cval{high}}$};
    
    \draw[-,thick,densely dotted] (token1) edge (token1c);
    \draw[-,thick,densely dotted] (token2) edge (token2c);
    \draw[-,thick,densely dotted] (token3) edge (token3c);
    \draw[-,thick,densely dotted] (token4) edge (token4c);
    
   \end{scope}

  \end{tikzpicture}
}
\caption{Example of a \cpn}
  \label{fig:cpnExplanation}
  \end{figure}

Definition~\ref{def:CPNs} provides a definition of a \cpn, which is a simplifying version of the original definition to keep the explanation simple. Yet, it covers all the cases necessary in this paper. It is worth highlighting that tokens  can also be associated with no values. To cover this case, we introduce the colorset $\bullet = \{\circ\}$, which namely corresponds to black tokens in normal Petri nets.
\begin{definition}[\cpn] 
A \cpn is a tuple $(P, T, A, \Sigma, V, C, N, E, G, I )$ where:
\begin{compactitem}
  \item $P,T,A$ are sets of places, transitions and direct arcs, respectively;
  \item $\Sigma$ is a set of color sets defined within the CPN model and $V$ a set of variables;
  \item $C: P \rightarrow \Sigma \cup \set{\bullet}$ is a \emph{color function} from places to a color set in $\Sigma \cup \{\bullet\}$;
  \item $N: A \rightarrow (P \times T) \cup (T \times P)$ is a \emph{node function} that maps each arc to either a pair $(p,t)$ indicating that the arc is between a place $p \in P$ to a $t \in T$, or $(t,p)$ indicating that the arc connects $t \in T$ to $p \in P$;
  \item $E: A \rightarrow V \cup \set{\var_\bullet}$ is an arc expression function,  assigning variables to arcs;
  \item $G: T \rightarrow \Expr(V)$ is a guard function that maps each transition $t \in T$ to an expression $G(t)$ with the additional constraint that $G(t)$ can only employ variables with which arcs entering $t$ are annotated: $G(t) \in \bigcup_{a \in A. N(a)=(p,t)} E(a)$;
  \item $I: P \rightarrow \MS(\Sigma \cup \set{\bullet})$ is an initialisation function assigning color values to places.
  For a place $p \in P$, $I(p)$ indicates the color of the tokens in $p$ at the initial marking, with $I(p) \in C(p)$.
\end{compactitem}
\label{def:CPNs}
\end{definition}
Variable $\var_\bullet$ is a special variable that is intended to only take on one value, namely $\circ$.
In general, for any arc $a \in A$, expression $E(a)$ can be more complex than just being a single variable. However, this simplification covers all the cases of arc's expressions we consider here. 
The concept of a marking $M$ can be easily extended to \cpn as $M: P \rightarrow \MS(\Sigma \cup \set{\bullet})$ where $M(p)$ is a multiset of elements, each of which it is the data (a.k.a.\ color in \cpn) associated to a different token in $p$.

A \cpn run is of the form $M^0 \goto{t^1,\cpnbind^1} M^1 \goto{t^2,\cpnbind^2} \ldots \goto{t^n,\cpnbind^n} M^n$ where $M^0 = I$ where, for all $1\leq i \leq n$, $\cpnbind^i : V \rightarrow (\Sigma \cup \set{\circ})$ is the so-called \emph{binding} function.
Function $\cpnbind^i$ is defined over the set of variables of the arcs entering transition $t^i$. When firing transition $t^i$ in marking $M^i$, only legal bindings are possible. A binding is legal for a transition $t$ if:\footnote{In the remainder, given a transition $t \in T$, we denote ${}^\bullet t=\{p \in P. \exists a \in A. N(a)=(p,t)\}$ and ${t}^\bullet=\{p \in P. \exists a \in A. N(a)=(t,p)\}$}
\begin{enumerate}
  \item Each variable $v$ associated with an arc $a$ s.t. $N(a)=(p,t^i)$ for some $p \in {}^\bullet t^i$ is in the domain of $\cpnbind^i$: $dom(\cpnbind^i)=\bigcup_{a \in A. N(a)=({}^\bullet t^i, t^i)} E(a)$
  \item $\cpnbind^i(v)$ takes on a value that is associated with one of the tokens in every place $p$ that has an arc to $t^i$ that is annotated with $v$: $\forall v \in dom(\cpnbind^i)$, $\forall p \in P$
  s.t. $\exists a \in A$ with $N(a)=(p,t^i)$ and $E(a)=v$, $\cpnbind^i(v) \in M^i(p)$.
  \item The guard of $t$ evaluates to true when variables are substituted as per $\cpnbind^i$: $\phi_{[\cpnbind^i]}=\true$
\end{enumerate}
Firing $t^i$ with $\cpnbind^i$ in marking $M^i$ leads to a marking $M^{i+1}$, denoted as $M^i \goto{t^i,\cpnbind^i} M^{i+1}$, that is constructed as follows:\footnote{Notation $a_{(p,t^i)}$ denotes the arc $a \in A$ s.t.\ $N(a)=(p,t^i)$ and cannot be employed if such an arc does not exist. Set-difference operator $\setminus$ is overridden for multisets: given two multisets $A$ and $B$, for each element $x \in B$ with cardinality $b_x>0$ in $B$ and cardinality $a_x \geq 0$ in $A$, the cardinality of $x$ in $B \setminus A$ is $\max(0,b_x - a_x)$; moreover, $x \not\in B \Rightarrow x \not\in B \setminus A$.}
\begin{equation*}
  M^{i+1}(p)=
  \left\{
  \begin{array}{ll}
    M^i(p) & \text{if } p \not\in ({}^\bullet t^i \cup t^{i \bullet}) \\
    M^i(p) \setminus [\cpnbind^i(E(a_{(p,t^i)}))] & \text{if } p \in {}^\bullet t^i \\
    M^i(p) \uplus [\cpnbind^i(E(a_{(t^i,p)}))] & \text{if } p \in t^{i \bullet}  \\
  \end{array}
  \right.
\end{equation*}
A firing $M^i \goto{t^i,\cpnbind^i} M^{i+1}$ is legal if $\cpnbind^i$ is a valid binding of $t^i$.
A \cpn run is legal if it is a sequence of legal firings.

\subsection{Translating \net{s} into Colored Petri Nets}
\label{sec:translation-general}


\begin{figure}[t]
\centering
\resizebox{.8\textwidth}{!}{
\input{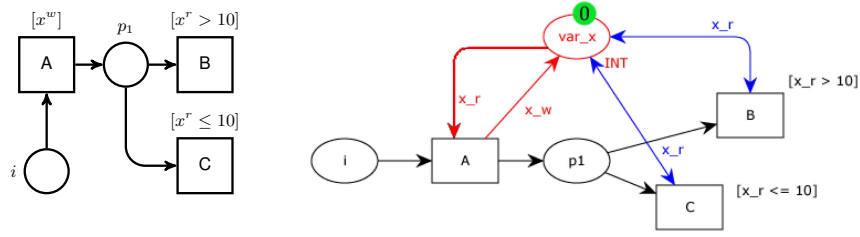}
}
\caption{Conversion of a simple \net to \cpn (left to right), where the green token indicates that the place contains a token with value 0. The color set of $var\_x$ is shown under the place. Arcs without annotations are equivalent to those with annotation $v_{\bullet}$ and places with no color sets are those associated with $\bullet$. Double-headed arcs are a shortcut to indicate that there are two arcs with the same inscription in either of directions.}
\label{fig:exampleMapping}
\end{figure}

This section illustrates how a \net $\anet = \DPNDELEM$  can be converted into a \cpn $\acpn = (\celem{P},\allowbreak  \celem{T},\allowbreak  \celem{A}, \allowbreak \celem{\Sigma},\allowbreak \celem{V},\allowbreak \celem{C}, \allowbreak\celem{N}, \allowbreak\celem{E}, \celem{G}, \celem{I} )$.
Intuitively, as exemplified in Figure~\ref{fig:exampleMapping}, the transitions and places of the \net become transitions and places of the \cpn. Each variable $\var$ of the \net becomes one \emph{variable place} that is associated with the same colorset as the variable type of $\var$ (place $var\_x$ in example in Figure~\ref{fig:exampleMapping} (right). These places always contain exactly one token, holding the current value of the variable. Guards are exactly the same as the guards of the \cpn, and if a transition writes a variable $\var$, the token in the variable place for $\var$ is consumed and a new token generated to model that is the value of $\var$ is updated. For instance, the fact that transition $A$ of the \net Figure~\ref{fig:exampleMapping} (left) writes a new value for variable $x$ (denoted $x^w$) is modelled in Figure~\ref{fig:exampleMapping} (right) through the two red arcs annotated with $x\_r$ and $x\_w$ that respectively enters and exits transition $A$: this allows the token holding the value of $\var$ to change value when returned back to the place. The read operations can be modelled as the blue arcs as in Figure~\ref{fig:exampleMapping} (right), with the same annotation, so that the token from the variable place is consumed and then put back. 
The initial marking of the \net becomes part of the initial marking of the \cpn: each variable place is initialized with a token that holds the initial value of the variable. In Figure~\ref{fig:exampleMapping} (right),  the place $var\_x$ contains a token with value $0$, assuming $\varInit(x)=0$.
The following formalizes this intuition.

\noindent
\textbf{Places.} The places of the \cpn consist of all places of the \net, plus one dedicated extra place $\xi(\var)$, hereafter called \emph{variable place}, for each \net variable $\var\in \delem{\Vars}$;
  $\celem{P} = \delem{P} \displaystyle \cup_{v \in \delem{\Vars}} \xi(v)$. A variable place $\xi(v)$ always has one token, and precisely the one holding the current value of variable $\var$ at each step of the simulation of the \cpn.\\
\textbf{Transitions.} The transitions of the \cpn and \net are the same: $\celem{T} = \delem{T}$.\\
\textbf{Arcs.} Each arc in $\delem{F}$ is preserved, and for any transition $t \in T$ and  variable read and/or written in $t$, we add two extra arcs: $\celem{A} = \delem{F} \cup \{ (t,\xi(v)), (\xi(v),t) \;|\; t \in \delem{T}, v \in \delem{\reads}(t) \cup \delem{\writes}(t)\}$,
  and the node function is defined as $\celem{N}(a)=a$ for any $a \in \celem{A}$.\\
  \textbf{Color sets.} The \cpn supports the same variable types as the \net, and we consider the color sets $\celem{\Sigma} = \set{ \Integers, \Reals, \Bools, \Strings, \UNITs }$ corresponding to the domains defined at the beginning of  Section~\ref{sec:syntax-semantics} for integers, reals, booleans and strings, respectively. 
  \textbf{Variables.} For each variable $v \in \delem{\Vars}$ the \cpn considers the variables $v^r$ and $v^w$, i.e.,
  $\celem{V} = \{ v^r, v^w | v \in \delem{\Vars} \} \cup \set{v_{\bullet}}$, where $v_{\bullet}$ is the special dummy variable with the only possible value $\circ$.\\
 \textbf{Color functions.} Recalling the shorthand notation $\typevar$ for typed variables in $\delem{\Vars}$, each place $p \in P_C$ is associated with a color set as follows. If $p\in\delem{P}$ then $\celem{C}(p) = \UNITs$, otherwise:
  \begin{equation*}
\celem{C}(p) = \left\{
\begin{array}{ll}
  \Integers & \text{ if there is } \var_{\Integers}\in\delem{\Vars} \text{ and } p=\xi(\var_{\Integers});\\
  \Bools & \text{ if there is } \var_{\Bools}\in\delem{\Vars} \text{ and } p=\xi(\var_{\Bools});\\
  \Strings & \text{ if there is } \var_{\Strings}\in\delem{\Vars} \text{ and } p=\xi(\var_{\Strings});\\
  \Reals & \text{ if there is } \var_{\Reals}\in\delem{\Vars} \text{ and } p=\xi(\var_{\Reals}).\\
  \end{array}
\right.
\end{equation*}\\
 \textbf{Guards.} Guards are not changed: $\celem{G}(t)=\delem{\guard}(t)$ for each $t \in \celem{T}$.\\
\textbf{Arc expressions.} The expression associated with any arc between a source node
    $s \in \celem{P} \cup \celem{T}$ and a target $t \in \celem{P} \cup \celem{T}$ with $(s,t) \in \celem{A}$ is as follows. If $(s,d) \in \delem{F}$ then $\celem{E}((s,d))=\var_\bullet$, otherwise:
  \begin{equation*}
\celem{E}((s,d)) = \left\{
\begin{array}{ll}
  v^r & \text{ if } d \in \delem{T} \text{ and } s=\xi(\var); \\
  v^r & \text{ if } s \in \delem{T} \text{ and } d=\xi(\var) \text{ and } \var\not\in \delem{\writes}(s);\\
  v^w & \text{ if } s \in \delem{T} \text{ and } d=\xi(\var) \text{ and } \var \in \delem{\writes}(s);\\
\end{array}
  \right.
\end{equation*}
\smallskip

  The first case refers to arcs of the \cpn that are also present in the original \net (e.g. in the set of arcs $F$); the places involved in these arcs contain tokens with no value associated and, which we represent by $\circ$, and thus the arcs are annotated with the $\var_\bullet$ variable. 
  The remaining cases refer to arcs connecting the variable places for each $\var\in\delem{\Vars}$ to a transition $t\in\celem{T}$. If $\var$ is written by $t$ then the incoming arc $(\xi(\var),t)$ and the outgoing arc $(t,\xi(\var))$ are annotated with $\varr$ and $\varw$, respectively. This allows the token holding the value of $\var$ to change value when returned back to $\xi(\var)$.
  If instead $\var$ is not written by $t$ then both arcs are annotated with the same inscription  $\varr$, guaranteeing that the value of token does not change. \\
  %
%
\textbf{Initialization.} Let $\delem{M_I}$ be the initial marking of the \net. Places that are also in the \net take on the same number of tokens as in the \net, whereas each variable place $\xi(\var)$ is initialized with a token holding the value specified by the initial \svassignment of the \net. Namely, $\celem{I}(p)=[\circ^{\delem{M_I}(p)} ]$ if $p \in \delem{P}$, i.e., $p$ is a place in the original net, otherwise $\celem{I}(p)= [ \delem{\varInit}(\var) ]$ where $p=\xi(\var)$.

This translation is correct: a \net is sound if and only if its translation to \cpn is sound, and it also allows one to  leverage on standard  techniques~\cite{Aalst1998workflow}. Specifically, in this paper we resort on building and analysing the reachability graph of the \cpn,
on which the conditions as in Definition~\ref{def:properties} can be checked as follows: properties P1 and P2 can be assessed by verifying that, for any state of the reachability graph, it is possible to reach a final marking. Property P3 can be checked by assessing whether the reachability graph contains at east one edge for every transition.


\begin{example}
Figure~\ref{fig:CPNfromDPN-preDiscr} illustrates how the working example is translated into a \cpn. The red the green elements implements the operations of reading and updating (i.e.\ writing) of the variables $\varok$ and $\varamount$, respectively.
\end{example}
\begin{figure}[hp!]
  \centering
  \includegraphics[height=0.8\textwidth, angle=270]{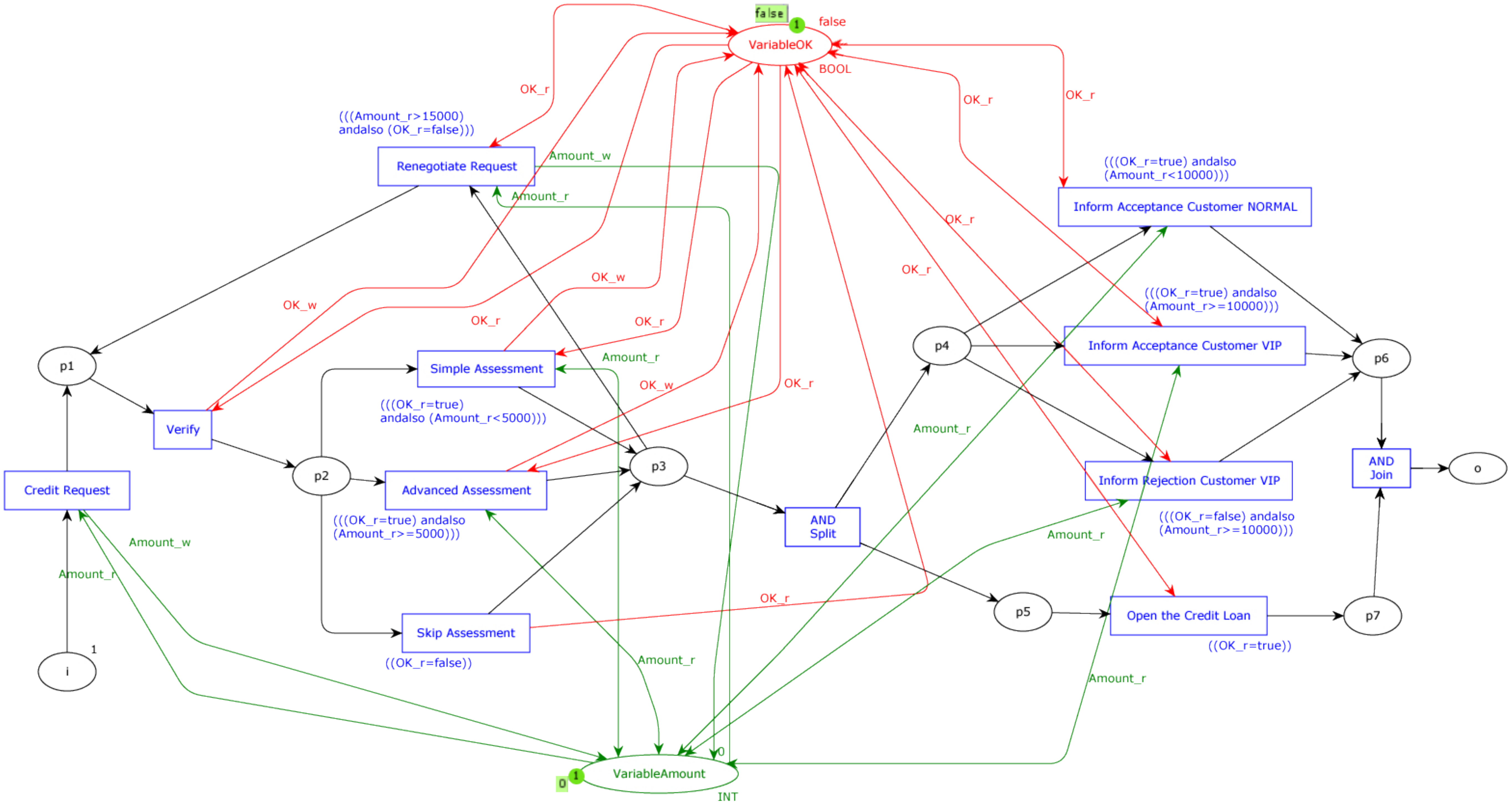}
  \caption{Translation of \net in Figure~\ref{fig:net} as \cpn. This \cpn is internally generated by our plug-in. Double-headed arcs are a shortcut to indicate that there are two arcs with the same inscription in either of directions.
  }
  \vspace{-.1in}
  \label{fig:CPNfromDPN-preDiscr}
\end{figure}

\subsection{Taming Infinity via Representatives}
\label{sec:translation-representatives}

However, although the translation is correct, it is easy to see that the reachability graph can have infinitely many distinct states. There source of infiniteness is twofold: on the one hand, the original \net itself and thus the \cpn can have an unbounded number of tokens; on the other hand, the process variables of the original \net determine color sets in the \cpn over possibly infinite domains. While the former can be tackled by standard techniques, the infiniteness of the process data makes them ineffective, and it does not give us any insight on whether soundness can be actually verified and, if so, how.

\begin{definition}[Constants of the process]
The set of \emph{constants} $\Consts{\var}\subset\varDomain$ related to a typed variable $\typevar\in\Vars$ of a \net is defined as the set of all the values $k$ such that either $\varr \odot k$ or  $\varw \odot k$ appears in any guard of any $t\in T$, with $\odot\in \Preds$.
\end{definition}

Observing that each such set $\Consts{\var}$ is finite and ordered, for each of these  we partition the universe $\Universe$ into $2\times\card{\Consts{\var}}+1$ intervals of values in which $\Universe$ can be partitioned wrt $\var$, and for each elect a \emph{representative}, which can be chosen arbitrarily among the values in the interval.
To correctly handle the case in which the domain $\varDomain$ of a variable $\var$ has no minimal or maximal elements, we define the set $\Consts{\var}^+$ as $\Consts{\var}$ with either or both of these two elements added, when needed.
Hence the set of representatives for $\var\in\Vars$ is computed as:
$\Repres{\var}:=\{ x\in\varDomain \mid x\in \Consts{\var} \text{ or } x= pick(x_1,x_2) \text{ for consecutive } x_1,x_2\in\Consts{\var}^+\}$
%
where $pick$ is a deterministic function returning a representative value in the specified interval, excluding the endpoints.\footnote{For dense domains such as real numbers such intervals are always nonempty, whereas for non-dense domains they might be empty. In this case, we consider $pick$ undefined.}
For a given value $x$ in the original domain $\varDomain$, we denote its representative as $rep(x)$, namely $rep(x):=x$ iff $x\in \Consts{\var}$, otherwise $rep(x)=y$ implies both $y\in(x_1,x_2)$ and $x\in(x_1,x_2)$. For $\bot$, we define $rep(\bot):=\bot$.

Let $\Repres{} := \set{\Repres{{\var}_1},\ldots, \Repres{{\var}_q}}$. We define a \svassignment \emph{restricted to} $\Repres{}$ as a function $\varStateR : \Vars \rightarrow \cup_{\var}\Repres{\var}$, with the restriction that $\varStateR(\var)\in  \Repres{\var}$ for any $\var$.
Given a \svassignment $\varState$ on the original domain $\varDomain$ of any variable $\var\in\Vars$, we \emph{compute} its restriction as
$\restrict{\varState}(\var) := rep(\varState(\var))$.


By considering a finite number of representative values, we can verify the soundness of a \net by checking the soundness of the corresponding \cpn if one restricts the values which can be assigned to each $\var$ to the set $\Repres{\var}$.
As we are going to show, although this cannot guarantee  the reachability graph of the \cpn to be finite-state in general, it suitably eliminates the infiniteness originating from the process data.
To ensure this, we need to add further constraints to the \cpn $\acpn$ as follows.
For each variable $\var \in \delem{\Vars}$ in the \net, we add an additional place $\rho(\var)$ to the set places $\celem{P}$ of the \cpn, which is meant to represent the restricted set of possible values of $\var$, namely $\Repres{\var}$. To this end, $\rho(\var)$ is assigned the same colorset as that of the variable place $\xi(\var)$, and it holds one token for each possible representative value in $\Repres{\var}$. This is achieved through the initialisation function of $\acpn$, by imposing $I(\rho(v))=\uplus_{x \in \Delta_v} [x]$. Then, for any transition $t \in \celem{T}$ and for each variable $\var \in \celem{\writes}(t)$, the representative value held by one token in $\rho(\var)$ is used to update the value of the token in the variable place $\xi(v)$.

\begin{figure}[p]
\centering
\includegraphics[height=0.7\textwidth,angle=270]{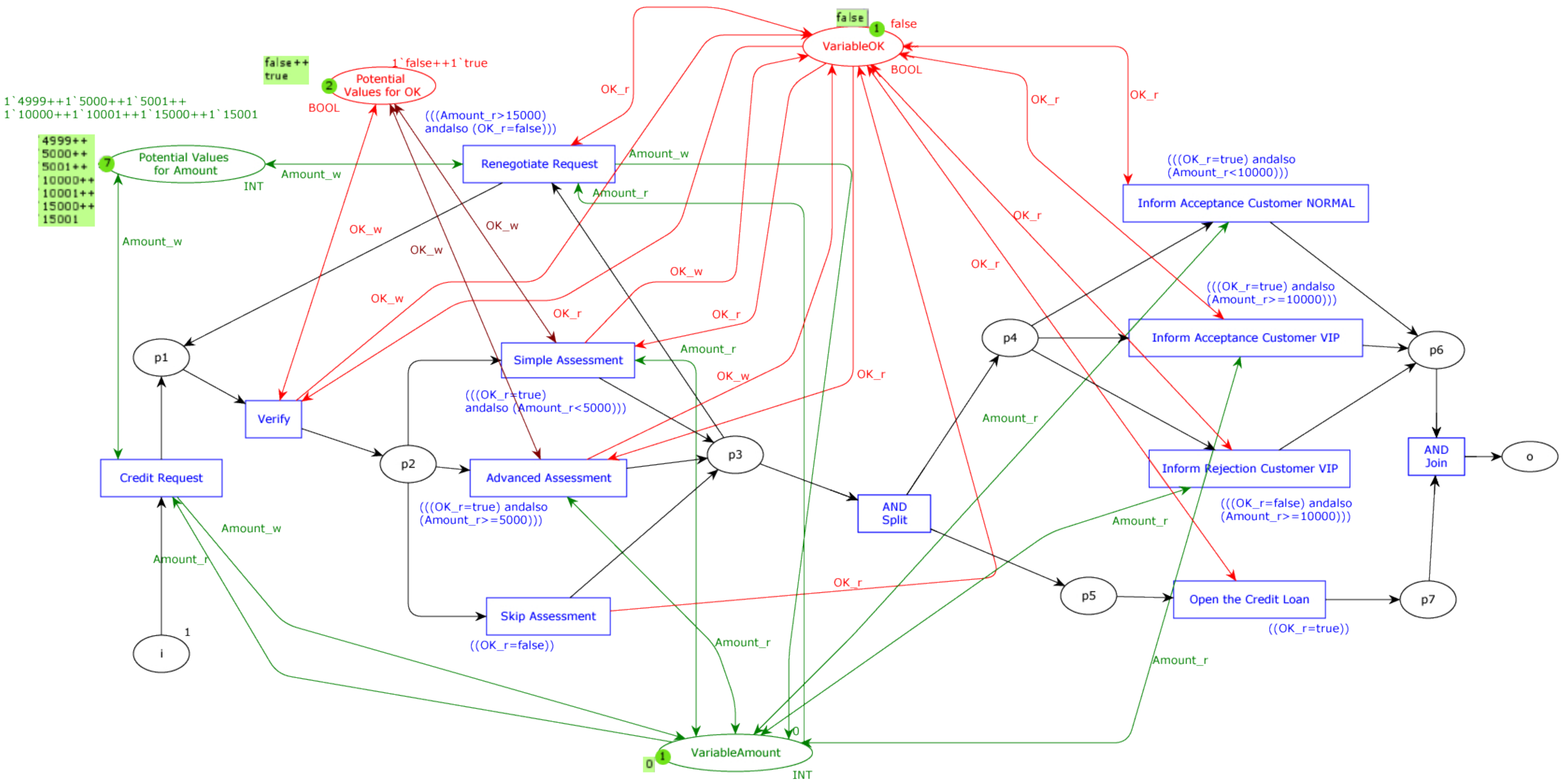}
  \caption{The extension of the \cpn in Figure~\ref{fig:CPNfromDPN-preDiscr} using representative values. Transitions also present in the original \net (black edges) are removed, and double arcs are a shortcut to indicate that there are two arcs with the same inscription in either of directions. }
  \label{fig:CPNfromDPN}
\end{figure}

\begin{example}
Consider, e.g., the transition \emph{Credit Request} in the model in Figure~\ref{fig:CPNfromDPN-preDiscr}.
This transition writes the integer variable $\varamount$. If we inspect all the guards, it is easy to see that the set of constants related to $\varamount$ is the set $\Consts{\varamount} = \set{5000, 10000, 15000}$, from which we select the set of representatives $\Repres{\varamount}=\set{4999,5000,5001,10000,10001,15000,15001}$ by including an \emph{arbitrary} value for each interval (e.g., in this case, $5001$ was arbitrarily chosen to represent all the values in the interval $(5000,10000)$), and a token for each element of $\Repres{\varamount}$ is created in $\rho(\varamount)$. As it can be seen in Figure~\ref{fig:CPNfromDPN}, which depicts the resulting \cpn $\acpnrep$, $\rho(\varamount)$ is called \textit{Potential values for Amount} and its tokens can be used as possible values for the variable $Amount\_w$, which the transition \emph{Credit Request} produces in the variable place $\xi(\varamount)$, there called \textit{VariableAmount}.
\end{example}

More formally, we add two arcs to $\celem{A}$: an arc $(t,\rho(v))$ from $t$ to the newly-introduced place $\rho(v)$ and a second arc $(\rho(v),t)$,
and define the expression function $\celem{E}$ so that, for transition $t \in \delem{T}$ and $\var \in \delem{\writes}(t)$, $\celem{E}((\rho(\var),t))=\celem{E}((\rho(\var),t))=v^w$.

\subsection{Correctness of the Translation}
\label{sec:translation-correctness}


We now discuss the correctness of the approach, showing that the \cpn defined above preserves the soundness properties of the original \net.
In order to do so, we could show that the \cpn built in the previous section (namely obtained by translating the original \net $\anet$ first into a \cpn $\acpn$ and thus into $\acpnrep$) preserve soundness. This would be hard, as it implies not only comparing a \net with a \cpn that is syntactically very different but also handling infinite domains for case variables. Instead, we first restrict the traces of $\anet$ so as to describe a new \net $\anetrep$ through the formal notion of restriction of variable domains introduced before, and only then show that there is indeed a tight relationship between such $\anetrep$ and the \cpn $\acpnrep$ introduced in the last section.
This is illustrated in Figure~\ref{fig:approach}.

\begin{figure}[]
\centering
\resizebox{.33\textwidth}{!}{

\begin{tikzpicture}[node distance = 3cm]

\node[] (0) {$\net$ $\anet$};
\node[right of=0] (1) {$\cpn$ $\acpn$};
\node[below of=0, yshift=1.9cm] (2) {$\net$ $\anetrep$};
\node[right of=2] (3) {$\cpn$ $\acpnrep$}; 
 
\draw[<->] (0) edge (1);
\draw[<->, very thick] (2) edge (3);
\draw[<->, very thick] (0) edge (2);
\draw[<->] (1) edge (3);
 
\end{tikzpicture}
}
  \caption{An intuitive diagram depicting our approach. First, we want to prove that $\anet$ is (data-aware) sound iff $\acpn$ is (data-aware) sound. However, to tame  infiniteness, we will show that $(i)$ $\anetrep$ constitutes a correct abstraction of $\anet$, namely that they are trace-equivalent, and then $(ii)$ that a similar property holds between $\anetrep$ and $\acpnrep$. As a result, we can analyse $\acpnrep$ in order to assess the soundness of $\anet$. We will then argue that this conclusion is in fact not limited to data-aware soundness.}
\label{fig:approach}
\end{figure}
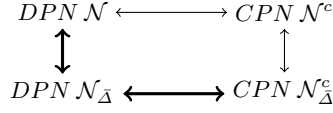

In the previous section we have shown how any state assignment $\varState$ can be restricted to the state assignment $\restrict{\varState}$ which only selects representatives from $\Repres{\var}$ for each variable $\var\in\Vars$.
Intuitively, this allows us to redefine the domain associated to $\var$ as $\Domain:=\tup{\Repres{\var},\Preds}$, and to consider a new \net $\anetrep$ that is as the original \net $\anet$ but in which states are of the form $(M, \restrict{\varState})$, including the initial state $(M_I,\restrict{\varState_I})$.
Similarly, transition firings are as in $\anet$, with the  difference that given a transition firing $(M,\varState) \goto{t,\bind} (M',\varState')$ in $\anet$, the state assignment for any $\var\in\Vars$ is $\restrict{\varState}'(\var)=\restrict{\varState}(\var)$ if $\varname\not\in \writes(t)$ and $\restrict{\varState}'(\var)=rep(\bind(\var))$ otherwise.

The following theorem shows that the abstraction step from $\anet$ to $\anetrep$ depicted in Figure~\ref{fig:approach} preserves trace-equivalence.

\begin{theorem}
\label{thm:first}
Given a \net $\anet$, then $\anetrep$ is trace-equivalent to $\anet$.
\label{theorem:anet_anetrep_equiv}
\end{theorem}

The intuition behind this result is that for any possible \emph{legal} process trace $\sigma_1=\tup{(t^1,\bind^1),\ldots, (t^n,\bind^n)}$ of $\anet$ we compare the run
$\tau_1= (M_I,\varInit) \goto{\sigma_1} (M_1^n,\varState^n)$
of $\anet$  with the run
$\tau_2= (M_I,\restrict{\varInit}) \goto{\sigma_2} (M_2^n,\restrict{\varState}^n)$
of $\anetrep$, with $\sigma_2=\tup{(t^1,\restrict{\bind}^1),\ldots, (t^n,\restrict{\bind}^n)}$, so that these runs are trace equivalent, where $\restrict{\varState}^i$ and $\restrict{\bind}^i$ are the restrictions of $\varState^i$ and $\bind^i$ with respect to $\Repres{}$, for each $i\in[1,n]$. Then, we show how for every legal transition firing $(t,\bind)$ from the last state of $\tau_1$, its restriction $(t,\restrict{\bind})$ is legal from the last state of $\tau_2$ and that $\tau_1\goto{(t,\bind)}(M',\varState')$ is still trace-equivalent to $\tau_2\goto{(t,\restrict{\bind})}(M'',\varState'')$.
The proof proceeds by induction on the length of runs, as the claim trivially holds for runs of length $0$: it is easy to show that if the claim does not hold then either $\tau_1\goto{(t,\bind)}(M',\varState')$ is not legal or it is indeed possible to simply consider the restriction of $\bind$ and thus obtain a run $\tau_2$ and $\tau_1\goto{(t,\bind)}(M',\varState')$ for which the claim holds. Similarly for the other direction.

\begin{proof}
First, observe that the claim holds for runs of length $0$.  Then, assume it holds also for runs of length $n$ but not for length $n+1$, and consider $\tau_1$ and $\tau_2$ as above.
Then it means that for some legal firing $(t,\bind)$ from the last state of $\tau_1$ we cannot find a legal firing $(t,\restrict{\bind})$ from the last state of $\tau_2$ so that the condition above holds. Or vice-versa. We address these two cases. First, Recall that trace-equivalent runs traverse the same markings, which implies $M_1^n=M^n_2$, hence we simply write $M^n$.

(1): If $(t,\bind)$ is legal from $\tau_1$ but $(t,\restrict{\bind})$ is not legal from $\tau_2$, then either $(a)$ $\restrict{\bind}$ cannot be defined as it is not possible to select a representative for $\bind(\var)$ for some $\var\in\writes(t)$ or $(b)$ the transition firing $(t^{n+1},\restrict{\bind}^{n+1})$ is not legal from $(M^n,\restrict{\varState}^n)$. If $(a)$ is true then $\bind^{n+1}(\var)=x$ but $x\not\in\Repres{\var}$. By definition, this means that $x\not\in\Consts{\var}$ and that there are no two consecutive $x_1,x_2\in\Consts{\var}$ with $x\in(x_1,x_2)$. This is only possible if the open interval $(x_1,x_2)$ is empty, which contradicts the fact that $\bind^{n+1}(\var)=x$. If $(b)$, by definition we have that, considering $\phi=\guard(t)$, either $(b1)$ $\phi_{[\restrict{\bind}^{n+1}]}$ is not satisfied, or $(b2)$ it is not true that $M^{n} \goto{t^{n+1}} M^{n+1}$. If $(b1)$ is true then for some $\var\in(\VarsR\cup\VarsW)$ either $\phi=\var$ and $\restrict{\bind}^{n+1}(\var)=\bot$ or $\phi$ is of the form $\var\odot k$ but it is not the case that $\odot(x,k)$ for $x=\restrict{\bind}^{n+1}(\var)$. In the former case it follows that also $\bind^{n+1}(\var)=\bot$ and $\phi_{[\bind^{n+1}]}\neq\true$, while the latter is not possible because if $x$ is a constant in $\Consts{\var}$ then it must be $\phi_{[\bind^{n+1}]}\neq \true$ as $x=rep(x)$, whereas if it is a representative value in the interval $(x_1,x_2)$ then there must exist another constant in $\Consts{\var}\in(x_1,x_2)$ that is also in $\reads(t)$, which implies that these are not consecutive constants, unless of course it is also false that $\odot(\bind^{n+1}(\var), k)$. Finally, if ($b2$) is true then $\tau_2$ is not legal.

(2:) The proof for the other direction is analogous, proving that for every legal firing $(t,\restrict{\bind})$ from the last state of $\tau_2$ in there exists a legal firing $(t,\bind)$ from the last state of $\tau_1$ so that $\tau_1\goto{(t,\bind)}(M,\varState')$ is still equivalent to $\tau_2\goto{(t,\restrict{\bind})}(M,\varState'')$.
\end{proof}

\medskip
We now address the relationship between $\anetrep$ and the corresponding \cpn $\acpnrep$, built from $\anetrep$ following the construction illustrated in Section~\ref{sec:translation-representatives}. As depicted in Figure~\ref{fig:approach}, our goal is to show that  the \cpn $\acpnrep$ captures all and only the possible runs of $\anetrep$, so that we can indeed analyse $\acpnrep$ in ProM, as explained in the next section, to verify properties of the original \net $\anet$.

First, with a little abuse of terminology, we extend the notion of trace-equivalence also for comparing \net{s} and \cpn{s}.

\begin{definition}[Trace-equivalence between \net{s} and \cpn{s}]We say that a \net run $\tau=(M_I,\varState_I)\goto{t^1,\bind^1}(M^1,\varState^1)\goto{t^2,\bind^2}\cdots\goto{t^n,\bind^n}(M^n,\varState^n)$ is \emph{trace-equivalent} to a \cpn run $\tau_c=M_{Ic}\goto{t_c^1,\gamma^1}M_c^1\goto{t_c^2,\gamma^2}\cdots\goto{t_c^n,\gamma^n}M_c^n$ iff $t^i=t_c^i$ for each $i\in[1,n]$, namely if the runs perform the same transitions ($M_{Ic}$ is the initial marking of the \cpn).
\end{definition}

Similarly, a \net $\anet$ and a \cpn $\acpn$ are trace-equivalent if $(a)$ for every legal run $\tau$ of $\anet$ there exists a trace-equivalent run $\tau_c$ of $\acpn$ and $(b)$ vice-versa.

\begin{theorem}
For any \net $\anet$, the corresponding $\anetrep$ and the \cpn $\acpnrep$ defined as in the previous section are trace-equivalent.
\label{th:equiv_dpn_cpn}
\end{theorem}

It is easy to see that the extra place $\rho(\var)$, included in $\acpnrep$ for any $\var\in\Vars$, makes it possible to generate a reachability graph composed of runs that do not distinguish distinct values of variables as long as these are represented by the same representative value in $\Repres{\var}$. However,
one main difference in addressing the trace-equivalence between a run of a \net $\anetrep$ and a run of the corresponding \cpn $\acpnrep$, with respect to the same notion between two \net{s} as for the previous theorem, is that their respective markings are structurally  different. In $\anet$ markings are of the form $M\in\MS(P)$ whereas for $\acpnrep$ they are computed by a function of the form $M_c: P \rightarrow \MS(\Sigma \cup \set{\bullet})$.
We thus need to define a correspondence relation between these two notions of markings, and more precisely between states of a \net (the same applies to $\anet$ and $\anetrep$) and markings of the corresponding \cpn $\acpnrep$, defined as follows. Given a state $(M,\varState)$ of a \net, for any place $p\in P$ we have that $M_c(p)=[\circ^{M(p)}]$ and for any $\var\in \Vars$ we have $M_c(\xi(\var))=[\varState(v)]$. Similarly, given a marking $M_c$ of $\acpnrep$, the corresponding state $(M,\varState)$ of the \net is so that for any place $p\in P$ we have $M(p)=\card{M_c(p)}$ and for any $\var\in\Vars$ we have $\varState(\var)=M_c(\xi(\var))$.
With such a notion of correspondence at hand, which we denote by writing $(M,\varState)\leftrightsquigarrow M_c$, we prove our result.

\begin{proof}
(1): We proceed by induction on the length of runs $\tau_1=(M_I,\restrict{\varState_I}) \goto{\sigma_1} (M^n,\restrict{\varState}^n)$ of $\anetrep$, with $\sigma_1=(t^1,\restrict{\bind}^1)\cdots(t^n,\restrict{\bind}^n)$.
First, it is easy to see that $(M_I,\varState_I)\leftrightsquigarrow M_{Ic}$ by construction (see Section~\ref{sec:translation-general}). Then, assume that $(M,\varState)\leftrightsquigarrow M_c$ and that there exists a legal transition firing $(t,\bind)$ from such state so that $(M,\varState)\goto{t,\bind}(M',\varState')$ but there is no firing $(t,\gamma)$ of $\acpnrep$ such that $M_c\goto{t,\gamma}M'_c$ and $(M',\varState)\leftrightsquigarrow M'_c$. Then, either ($a$) any possible $(t,\gamma)$ is not legal from $M_c$ or ($b$) it is the case that $(M',\varState)\not\leftrightsquigarrow M'_c$. The former case is not possible. First, $t$ must be enabled in $M_c$ because $M_c(p)=[\circ^{M(p)}]$ for every $p\in P$ and $M_c(\xi(\var))=[\varState(v)]$ for any $\var\in \Vars$, with arcs between $t$ and places $\xi(\var)$ and $\rho(\var)$ as described in the definition of $\acpnrep$. Second, we can pick $\gamma$ such that $\gamma(\var)\in M_c(\xi(\var))$ for all $\var=E(a)$, which is equivalent to say that $\gamma(\var)=\varState(\var)$ for read variables $\var\in \reads(t)$ and $\gamma(\var)=\bind(\var^w)$  for written variables $\var\in \writes(t)$. This must be possible or otherwise $(t,\bind)$ would not be legal. Third, and for the same reason, if such binding $\gamma$ is so that the guard $\phi$ is not satisfied, then the same would hold for $\bind$ (which agrees with $\varState$ for read variables, namely $\bind(\var^r)=\varState(\var)$ for each $\var\in \reads(t)$), which would imply that $(t,\bind)$ is not legal from $(M,\varState)$. Since $(M,\varState)\leftrightsquigarrow M_c$ it follows that ($b$) is also not possible, as $M\goto{t}M'$ and the binding $\gamma$ as above is consistent with $\varState$ and $\bind$: specifically, for places $p\in{}^\bullet t$ we have $M'_c(p)=M_c(p)\setminus [\gamma (E(a_{(p,t)}))]$, and for places $p\in t^\bullet$ we have $M'_c(p)=M_c(p)\uplus [\cpnbind(E(a_{(t,p)}))]$, where $\gamma$ agrees with the marking $M_c$, namely $\gamma(E(a_{p,t}))$ selects a value from $M_c(p)=[\circ^{M(p)}]$ when $p\in P$, and $\gamma(E(a_{p,t}))$ selects a value from $M_c(p)=[\varState(\var)]$ when $p = \xi(\var)$, with  $\bind(\var^r)=\varState(v)$ for $\var\in\reads(t)$.

(2): The type of reasoning for the other direction is analogous.
\end{proof}

\medskip
Putting the theorems together, as depicted intuitively in Figure~\ref{fig:CPNfromDPN}, the following theorem holds, as the property of trace-equivalence between runs is clearly transitive.

\begin{theorem}
For any \net $\anet$, the corresponding \cpn $\acpnrep$ obtained with the construction so far described is trace-equivalent to $\anet$.
\label{theorem:anetrep_acpnrep_equiv}
\end{theorem}

We thus turn to the question of determining whether the property of trace equivalence stated in the previous theorems allow one to transfer interesting properties. However, note that we cannot evaluate the same properties on both $\anet$ and $\acpnrep$ without a translation for \cpn{s}. The translation is intuitive although cumbersome, and it is left to the reader (an example is provided in the next section).

\begin{corollary}
The soundness properties in Definition~\ref{def:properties} hold in $\anet$ iff their translations hold in $\acpnrep$, hence $\anet$ is sound iff $\acpnrep$ is sound.
\end{corollary}

The proof is quite intuitive, and relies on observations already used in the proofs above. By the proof of the Theorem~\ref{theorem:anet_anetrep_equiv}, it follows that the reachability graphs of $\anet$ and $\anetrep$ have the same branching structure with respect to transitions in $T$: at each step, any task that is enabled in $\anet$ in a given run fragment is also enabled in the corresponding trace-equivalent run in $\anetrep$. Similarly, by the proof of Theorem~\ref{theorem:anetrep_acpnrep_equiv}, trace-equivalent runs of $\anetrep$ and $\acpnrep$ traverse markings that are in correspondence, as defined earlier. Therefore, at every step,  every enabled transition in the $\acpnrep$ must be enabled in the corresponding state of the $\acpnrep$.

\begin{example}
Consider again the running example from Example~\ref{ex:example1} depicted in Figure~\ref{fig:net}. As already anticipated (see Example~\ref{ex:unsound}), the net is unsound as it is possible to reach a deadlock (e.g. having a token in $p5$ and $p6$). Note that the same happens in its resulting representative \cpn which uses representative values, shown in Figure~\ref{fig:CPNfromDPN}. The key point is that, while in the original \net $\anet$ there exists an infinite number of values (e.g. for the requested credit) for which a deadlock can be reached, all these are correctly represented by a finite number of runs in  $\acpnrep$.
\label{ex:running_ex_end}
\end{example}

\subsection{Relating Data-aware soundness and Decision-aware soundness}
\label{sec:data-vs-decision}

The previous result suggests that our technique is in fact not limited to data-aware soundness, but it can be applied to any property that does not rely on specific identity values of case variables (because $\acpnrep$ is insensible to the assignments of variables, as it uses representative values). We consider here the properties in \cite{BaHW17} which characterise various notions of decision-aware soundness.

\begin{corollary}
The truth value of all the properties in \cite{BaHW17} is the same in $\anet$ and $\acpnrep$.
\end{corollary}

The rest of this section is devoted to support this claim. First, we now look at the decision-aware properties from \cite{BaHW17}, adapting them to our case, where DMN decisions are represented in a \net as transitions that are used to model the business rule task that is associated to the decision table (we consider here a \emph{unique} hit policy for decision tables, see Section~\ref{sec:dmn}).

\newcommand{\Decs}{Decs}

We start by addressing the two fundamental properties of conditional completeness and conditional output coverage. As shown in Section~\ref{sec:encoding}, a decision fragment is modelled in a \net as a set of transitions, as exemplified in Figure~\ref{fig:dmn-net}. We denote one such set as $\DecTrans$, and all these decisions as $\Decs$. This notation allows us to refer to all the transitions that belong to the same decision.

Also, for convenience of notation, given a \net $\anet$ we consider the set $Reach(\anet)$ defined as the set $\set{ (M,\varState)~|~ (M_I,\varState_I) \goto{*} (M,\varState) }$ for a given \net $\anet$, that is, the set of all  states reachable through \emph{legal} runs, which are intrinsically data-aware.


\medskip
\noindent
\emph{Conditional completeness (P4).} As explained in Section~\ref{sec:dmn}, decision fragments are captured by transitions with guards encoding the rules of decision tables (see for instance Figure~\ref{fig:dmn-net}). A set of transitions, one for each rule in a table, is therefore said to be conditionally complete iff at least one can legally fire from any reachable state in which it is enabled my the marking. This is captured by the property $\forall \DecTrans \in \Decs, (M,\varState)\in Reach(\anet)$. $\exists t\in \DecTrans.$ $(M\goto{t}M' \Rightarrow \exists t'\in\DecTrans\bind, \varState'.$ $(M,\varState)\goto{t',\bind}(M',\varState'))$. The property simply states that from any reachable state from which transitions corresponding to a decision task are available, then for each decision there exists at least one transition from the same set that is legally executable. It however immediate to see that $P4$ is implied by $P1$, which prescribes that from any reachable state it is always possible to complete the process, irrespective of the \svassignment{s} and thus irrespective of the values of variables that transitions produce along the run.


\medskip
\noindent
\emph{Conditional  output coverage (P5).} The property holds for a decision table when all outputs are covered by  conditions in the succeeding gateway. In our formalism, this correspond to check that for every possible legal transition firing $(t,\bind)$, from the last state of a legal run fragment, we have that the resulting state is not a dead-end. Hence: $\forall (M,\varState)\in Reach(\anet)$. ($\exists \bind,\varState'$. $(M,\varState)\goto{t,\bind}(M',\varState') \land M'\neq M_F) \Rightarrow \exists (M'',\varState''). (M',\varState')\goto{*}(M'',\varState'')$.  The property is  true in $\anet$ iff it is true for all $t\in \DecTrans$. Again, note that this property is implied by $P1$.


\medskip
\noindent
\emph{State-Based Decision Deadlock Freedom.} It requires that  every transition $t\in \DecTrans$, such that $(t,\bind)$ is legal from the last state of a legal run fragment, is conditionally complete and all its outputs are conditionally covered. Hence, the property corresponds to checking $P4$ and $P5$.

\medskip
\noindent
\emph{State-based Dead Branch Absence.} It requires that along any branch in each decision fragment there is an output of the decision model such that this branch is selected for execution. As we represent a transitions both the business rule task and the following conditions, this is a relaxation of $P3$, limited to those transitions that follow one in $\DecTrans$.

\medskip
\noindent
\emph{Decision-Aware Soundness.} A \net is decision-aware sound iff it is classically sound, it is decision deadlock free and state-based dead branch free. Note that, under such conditions, the three properties of classical soundness \cite{Batoulis2017} are guaranteed also with respect to the possible decision-aware paths of the process. In our setting, this  corresponds to $P1$-$P3$ together with $P4$ and $P5$. 
However, recall that $P1$ implies $P4$ and $P5$, and hence this definition is equivalent to data-aware soundness (Definition~\ref{def:properties}).


\medskip
\noindent
\emph{Decision-Aware Relaxed Soundness.} It requires relaxed soundness, state-based decision deadlock freedom and state-based dead branch absence. The first specifies that  every transition can participate in at least one sound
firing sequence, which is captured by $P3$ and in out setting already implies state-based dead branch absence. The property thus corresponds to $P3$, $P4$, $P5$.

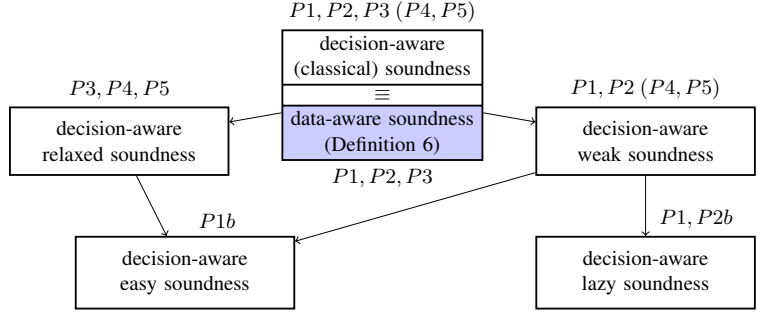
\begin{figure}[t!]
\centering
\resizebox{.7\textwidth}{!}{

\begin{tikzpicture}[node distance = 4cm]

\tikzstyle{place}=[rectangle,thick,draw=black,  inner sep=6pt]
\tikzstyle{rsplit}=[thick,draw=black, rectangle split,
       rectangle split parts=3,
       draw=black, 
       minimum height=2em,
       text width=3cm,
       rectangle split part fill={white!30,white!30,blue!20},
       inner sep=2pt,
       text centered]

\node[rsplit] (0) [text width=2.9cm,align=center, label={90:$P1,P2,P3$ $(P4, P5)$}, label={270:$P1,P2,P3$}] {decision-aware\\ (classical) soundness \nodepart{two} $\equiv$ \nodepart{three} data-aware soundness  (Definition~\ref{def:properties})};
\node[place, left of=0] (2) [yshift=-.7cm, text width=2.9cm,align=center, label={90:$P3, P4, P5$}] {decision-aware\\ relaxed soundness};
\node[place, right of=0] (3) [yshift=-.7cm, text width=2.9cm,align=center,  label={90:$P1,P2$ $(P4, P5)$}] {decision-aware\\ weak soundness};
\node[place, below of=2, xshift=1cm,yshift=2cm,  label={80:$P1b$}] (4) [text width=2.9cm,align=center] {decision-aware\\ easy soundness};
\node[place, below of=3, yshift=2cm,  label={80:$P1,P2b$}] (5) [text width=2.9cm,align=center] {decision-aware\\ lazy soundness};

\draw[->] (0) edge (2);
\draw[->] (0) edge (3);
\draw[->] (2) edge (4);
\draw[->] (3) edge (5);
\draw[->] (3) edge (4);
 
\end{tikzpicture}
}
  \caption{Relationship between data-aware soundness and the various notions of decision-aware
soundness, as illustrated in \cite{BaHW17}. The figure also shows to which properties these correspond, once translated for \net{s}.
}
\label{fig:soundnesses}
\end{figure}

\medskip
\noindent
\emph{Decision-Aware Weak Soundness.} It requires classical weak soundness, which is less strict than decision-aware soundness as it allows dead transitions (although any transition that can fire must always lead to a proper termination) and state-based decision deadlock freedom. As weak soundness is captured by $P1$ and $P2$, the property corresponds to checking $P1$, $P2$, $P4$ ($P5$ is implied).

\medskip
\noindent
\emph{Decision-Aware Lazy Soundness.} It allows the net to be lazy in the sense that there can be tokens left in the net after a token appeared on the final place (i.e. violating $P2$). When considering our execution semantics, it requires that from every (legally) reachable state the end transition (i.e., leading to the output place $p_{out}$ of $\anet$) can be reached exactly once. We can express this as
the conjunction of $P1$ with a simple relaxation of $P2$, namely $P2b:$ $\forall (M,\varState)\in Reach(\anet)$. $M(p_{out})\leq 1$.

\medskip
\noindent
\emph{Decision-Aware Easy Soundness.} It requires classical easy soundness and that at least one run (hence trace) to be state-based decision deadlock free and state-based dead branch free. We provide here the most reasonable interpretation of this definition. Since easy soundness simply requires that for at least one token on the initial place a token will eventually appear  on the final place, the requirement is captured by checking the existence of at least one run in which a relaxation of $P1$ holds ($P4$ and $P5$ are guaranteed by the legality of the selected trace). $P1b$: $\exists \sigma.$ $(M_I,\varInit) \goto{\sigma} (M_F,\varState)$.

\medskip
With arguments very similar to those applied for the proof of Theorem~\ref{theorem:anetrep_acpnrep_equiv}, it can be easily seen that these properties hold in $\anet$ iff they hold in $\acpnrep$ (once rewritten for \cpn{s}), due to trace-equivalence. Indeed, they are all based on legal runs and legally reachable states,  which  are all accounted for in $\acpnrep$. We exemplify such reasoning for $P4$, as the argument is akin to those for  $P1b$, $P2b$ and $P5$.

\begin{proof} (Sketch.)
 First, we express the corresponding requirement $P4_c$ on \cpn{s} for any set of transitions $\DecTrans$ as $\forall M_c.$ $(M_{Ic}\goto{*} M_c$ $\land$ $\exists t\in\DecTrans, \varState.$ $M_c\leftrightsquigarrow (M,\varState) \land M\goto{t}M') \Rightarrow  (\exists t'\in\DecTrans, \gamma, \varState'. $ $M_c\goto{t',\gamma} M'_c \land M'_c\leftrightsquigarrow (M',\varState'))$.
 Assume that $P4$ is satisfied in $\anet$ for every set $\DecTrans$. Then for every run $\tau=(M_I,\varState_I) \goto{\sigma} (M^n,\varState^n)$ in $\anetrep$ with $M^n\goto{t} M'$ for some $M'$ and $t\in\DecTrans$ we have $(M^n,\varState^n)\goto{t',\bind}(M',\varState')$ for some $t'\in\DecTrans$ and $\bind$, which agrees with $\varState^n$ by construction. Since they are trace-equivalent, this implies that there exists a trace $\tau_c=M_{Ic} \goto{*} M^n_c$  in $\acpnrep$ with $(M^i,\varState^i)\leftrightsquigarrow M^i_c$ for each $i\in[1,\card{\sigma}]$, which implies that for any place $p\in P$ we have $M^i(p)=\card{M^i_c(p)}$ and for any $\var\in\Vars$ we have $\varState^i(\var)=M^i_c(\xi(\var))$. Hence it must also be that $t'$ is enabled in $\acpnrep$, namely $M^n_c\goto{t',\gamma} M'_c$ for some $M'_c$ and $\gamma$ or otherwise the trace-equivalence would be violated. Therefore $P4_c$ is  satisfied for $\DecTrans$. A similar argument holds for the other direction.
\end{proof}

\bigskip
We conclude by noting that, shown in Figure~\ref{fig:soundnesses}, the same relationships between all the notions  presented in \cite{BaHW17} are correctly reconstructed, as expected. In particular, as shown in the figure, it turns out that our notion of data-aware soundness for \net{s} is equivalent to that of decision-aware soundness for processes associated with DMN decision models. That is, if we either (a) check decision-aware soundness on processes with DMN decisions or (b) translate them into a \net{s} and then check data-aware soundness, then we get the same result. This supports the claim that data-aware soundness in indeed suitable for capturing the soundness of Petri net-based process models enriched with case data and decisions. It requires the \emph{existence of at least one way} of completing a process  through the execution of a legal sequence of transition firings, \emph{in any possible case} that is allowed by both the control-flow and the data-flow of the process. This is a very strong property: it is guaranteed \emph{irrespective} of the values that may be written along the execution.

At the same time, note that data-aware soundness does not rely on the specific structure of the process: while decision-aware soundness, as studied in \cite{BaHW17}, assumes that decisions are represented as \emph{decision fragments}, thus fixing a specific shape of the process, data-aware soundness is evaluated on generic \net{s}. This allows one to consider processes with arbitrary conditions on the data they manipulate, which makes our approach more general. Moreover, the execution semantics defined here is itself data-aware, as opposed to imposing additional requirements on the decision tables, in addition to checking classical soundness.

\section{Implementation And Experiments}
\label{sec:implementation}

\begin{figure}[t!]
  \centering
  \includegraphics[width=\textwidth]{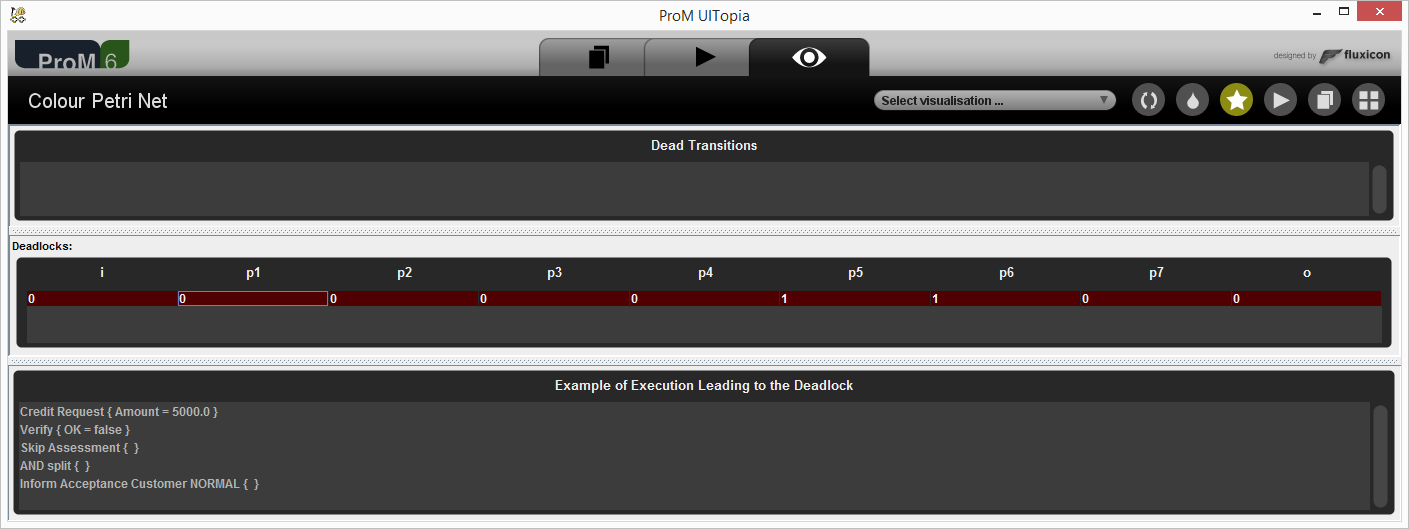}
  \caption{Screenshot of the tool that implements the soundness-checking technique described in this paper. The screenshot refers to the working example in Figure~\ref{fig:net}.}
  \label{fig:screenshotWE}
\end{figure}

\begin{figure}[t!]
  \centering
  \includegraphics[width=\textwidth]{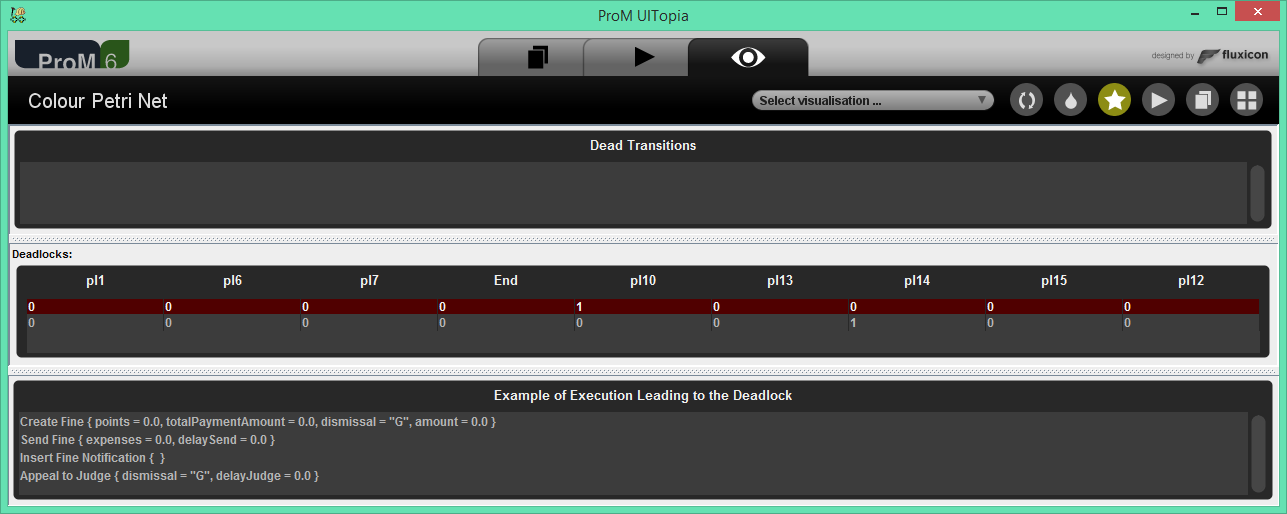}
  \caption{Screenshot of the tool that implements the soundness-checking technique described in this paper. The screenshot refers to the feedback on a data-aware model of a real-life process for road-traffic fine management.}
  \label{fig:screenshot}
\end{figure}

Our soundness-checking technique has been implemented as a Java plug-ins for ProM, an established open-source framework for implementing process mining algorithms and tools (see \url{\texttt{http://www.promtools.org/}}), which supports both the PNML and the BPMN file formats to load process models in those two formats.
ProM also implements numerous algorithms for discovering process models that integrate the decision perspective (e.g.~\cite{Leoni.2013b}). Thanks to this,  we can employ our technique to validate the soundness of models where the decision perspective is mined from event data and models can be expressed in the two mentioned notations.
In particular, the soundness-checking technique is available in the ProM \textit{nightly} build after ensuring that the ProM package \emph{DataPetriNets} is installed. The plug-in is named \emph{Compute Soundness of a Data Petri Net} and takes a \net as input.
Figure~\ref{fig:screenshotWE} refers to the output for the working example in Figure~\ref{fig:net}. The output illustrates the list of dead transitions as well as the undesired deadlocks, namely the list of markings in which no transitions are enabled although they are not final. For the working example the only undesired deadlock is the marking with one token in $p5$ and $p6$. Clicking on the deadlock, at the bottom the plug-in shows an example of execution that leads to that marking, namely when $Amount$ is 5000 and $Verification$ is \false. 

We performed a number of experiments with data-aware models of real-life processes that were used in previous publications and theses:
\begin{enumerate}
  \item We used the model of the real-life process for the management of road-traffic fines, which is illustrated in Figures 7 and 8 of~\cite{Mannhardt2016}. Space limitation prevents us from showing the models here. Figure~\ref{fig:screenshot} shows the feedback screen the ProM plug-in for soundness checking: no transition is dead and two deadlock markings can be identified. By clicking on any of the deadlock, an example of execution that leads to that deadlock is shown at the bottom. By inspecting the model, one can easily observe that the deadlock at the top (i.e.\ with a token in place $pl10$) is caused by the fact that transition \emph{Appeal to Judge} can assign any value to variable \emph{dismissal}. This transition is followed by a XOR split where two alternative transitions are possible, depending on the value of variable \emph{dismissal}: \texttt{NIL} or \texttt{\#}. However, the model does not impede \emph{Appeal to Judge} to assign other values, e.g.\ \texttt{G}, thereby causing a deadlock.
  \item We checked the soundness of the data-aware models reported in Figures 13.6 and 15.6 of the Ph.D.\ thesis by Mannhardt~\cite{FelixPhD}. Both models refer to processes that are executed within hospitals: the former is about curing patients with sepsis and the latter manages the hospital billing to patients. These models were partly hand designed and partly mined through process-discovery techniques.
\begin{figure}[p]
  \centering
  \includegraphics[height=\textheight]{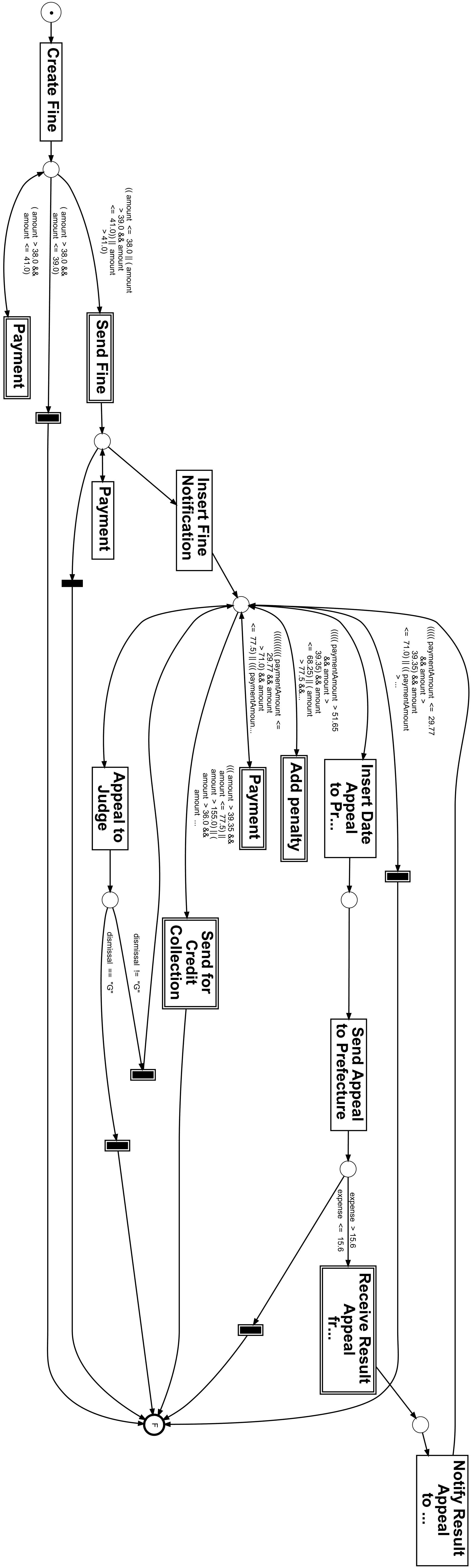}
  \caption{A data-aware Petri Net that models a road-traffic fine management process. The structure of the transitions and places of the net is manually designed and the guards are discovered through decision-mining techniques. The read and write operations are not shown and long guards are partly cut in order to maintain some readability.}
  \label{fig:discoveredModel}
\end{figure}

  \item We use the same model as at point 1 but, instead of keeping in the pre-existing guards, we employed the guard discovery technique discussed in~\cite{Mannhardt.2016b}, which does not formally guarantee that the discovered models comply the properties of Definition~\ref{def:properties}. The resulting model is in Figure~\ref{fig:discoveredModel}, which is data-aware sound. The analysis has not indeed reported dead transitions or deadlocks.
\end{enumerate}

The models at points 1 and 2 were analysed for deadlocks and dead transitions in a matter of seconds. The model at point 3 required 1.9 hours to return the analysis results. This difference is due to the fact that the model at point 3 is likely over-precise for what concerns the decisions. Therefore, the decisions are modelled through complex guards with several atoms; as a consequence, the search space to visit grows significantly.


\section{Conclusions}
\label{sec:conclusion}

In this paper we have introduced a holistic, formal and operational approach to verify the end-to-end soundness of \netname{s}, which we called data-aware soundness. Thanks to the solid formal foundation of \net{s}, we defined a notion of soundness for these nets to incorporate the decision perspective, and developed a technique for assessing such property that can be directly implemented on existing tools. We also characterised how our definition of data-aware soundness is related to known notions of decision-aware soundness in the literature.
In future work, we plan to address more sophisticated guard languages than the one considered in this paper, for instance by allowing to compare variables through guards such as $(\var_1^w \geq \var_2^r \land \var_1^w \neq \var_3^w)$. Note however that this goes beyond DMN S-FEEL and thus requires more sophisticated encoding techniques, although we believe this to be a decidable setting.
Further, we aim at extending our results to other data domains. This is a quite delicate task, since even minimal extensions may lead to undecidability. For instance, by enriching integer domains by a successor predicate, we immediately get an undecidability result for soundness, even in the simple case of \net{s}  with two case variables.
Finally, we  also have some intriguing ideas on how to optimize the technique presented in this paper. In its current form, nondeterminism is managed \emph{eagerly}, that is, by generated branches for possible values as soon as a variable is written. It appears instead promising to manage nondeterminism \emph{lazily}, i.e., by postponing such choice to the moment where the variable actually appears in a guard, hence considering sets of possible representatives at the same time. This would not preserve trace-equivalence, but could still preserve soundness.

\bibliographystyle{abbrv}
\bibliography{main-bib}

\end{document}